\documentclass[10pt,journal,twoside]{IEEEtran}

\IEEEoverridecommandlockouts

\usepackage{graphicx}
\usepackage{epstopdf}
\usepackage{array}
\DeclareGraphicsExtensions{.eps,.pdf} 
\graphicspath{ {figures/} }
\usepackage{cite}

\usepackage{amssymb,amsmath}
\usepackage{subfigure}
\usepackage{amssymb,amsmath,mathtools}
\usepackage{amssymb,amsmath}
\usepackage{algorithmic}
\usepackage{color}
\usepackage{multirow}
\usepackage{stackrel,amssymb,amsmath}
\DeclareMathOperator*{\argmax}{arg\,max}

\makeatletter
\newcommand\restartchapters{\par
  \setcounter{chapter}{0}%
  \setcounter{section}{0}%
  \gdef\@chapapp{\chaptername}%
  \gdef\thechapter{\@arabic\c@chapter}}
\makeatother

\newtheorem{theorem}{Theorem}

\newtheorem{remark}{Remark}

\newtheorem{lemma}{Lemma}

\makeatletter
\g@addto@macro\normalsize{%
 \setlength\abovedisplayskip{4pt}
 \setlength\belowdisplayskip{4pt}
 \setlength\abovedisplayshortskip{4pt}
 \setlength\belowdisplayshortskip{4pt}
}
\makeatother

\makeatletter
\def\endthebibliography{%
	\def\@noitemerr{\@latex@warning{Empty `thebibliography' environment}}%
	\endlist
}
\makeatother

\makeatother

\IEEEpubid{\makebox[\columnwidth]{\copyright~2019 IEEE. Personal use of this material is permitted. Permission from IEEE must be obtained for all other uses, in any current or future media, including reprinting/republishing this material for advertising or promotional purposes, creating new collective works, for resale or redistribution to servers or lists, or reuse of any copyrighted component of this work in other works \hfill} \hspace{\columnsep}\makebox[\columnwidth]{}} 

\begin{document}
\bstctlcite{IEEEexample:BSTcontrol}
\title{\huge Security and Reliability Analysis of Satellite-Terrestrial Multi-Relay Networks with Imperfect CSI}
\author{ \normalsize
\IEEEauthorblockN{$\text{Tan N. Nguyen}$,  $\textit{Member, IEEE}$, $\text{Dinh-Hieu Tran}, \textit{Graduate Student Member, IEEE}$, $\text{Trinh Van Chien}$,  $\textit{Member, IEEE}$, $\text{Van-Duc Phan}$, $\text{Miroslav Voznak}, \textit{Senior Member, IEEE}$, $\text{Symeon Chatzinotas}, \textit{Senior Member, IEEE},$ }
\thanks{Manuscript received XXX; revised XXX;
	accepted XXX. Date of publication XXX; date of
	current version XXX. This research was financially supported by Van Lang University, Vietnam, and was supported in part by the Ministry of Education, Youth and Sports of the Czech Republic under the grant SP2021/25 and e-INFRA CZ (ID:90140). (\textit{Corresponding author: Van-Duc Phan.})}
\thanks{Tan N. Nguyen is with the Communication and Signal Processing Research Group, Faculty of Electrical and Electronics Engineering, Ton Duc Thang University, Ho Chi Minh City, Vietnam. (e-mail:nguyennhattan@tdtu.edu.vn).}
\thanks{Dinh-Hieu Tran is  currently with Nokia Bell Labs, France, and was with the Interdisciplinary Centre for Security, Reliability and Trust (SnT), the University of Luxembourg, Luxembourg, (e-mail:dinh-hieu.tran@nokia.com).}
\thanks{Trinh Van Chien is with the School of Information and Communication Technology (SoICT), Hanoi University of Science and Technology (HUST), 100000 Hanoi, Vietnam, (e-mail: chientv@soict.hust.edu.vn).}
\thanks{Van-Duc Phan is at Faculty of Automotive Engineering, School of Engineering and Technology, Van Lang University, Ho Chi Minh City, Vietnam (email: duc.pv@vlu.edu.vn).}
\thanks{Miroslav Voznak is at VSB – Technical University of Ostrava, 17. listopadu 2172/15, 708 00 Ostrava, Czech Republic. (e-mail: miroslav.voznak@vsb.cz).}
\thanks{Symeon Chatzinotas is at the Interdisciplinary Centre for Security, Reliability and Trust (SnT), the University of Luxembourg, Luxembourg, (e-mail: symeon.chatzinotas@uni.lu).}
\thanks{\copyright~2021 IEEE. Personal use of this material is permitted. However, permission to use this material for any other purposes must be obtained from the IEEE by sending a request to pubs-permissions@ieee.org.}

%\vspace*{-0.45cm}
}
%$\text{Dinh-Hieu Tran}, \textit{Student Member, IEEE}$, and $\text{Chung D. Ho}$, 
\maketitle
\thispagestyle{empty}
\pagestyle{empty}
%\vspace*{-2cm}
\begin{abstract}
This work investigates the security and reliability analysis for a novel satellite-terrestrial (SatTer) network. Specifically, a satellite attempts to transmit confidential information to a ground user (GU) via the support of multiple relay nodes in the presence of an eavesdropper that tries to overhear the information. A friendly jammer is deployed to improve the secure transmission between the satellite and the relays. Furthermore, satellite-to-relay generalized Rician fading channels and imperfect channel state information (CSI) are deployed to examine a general system model.  In this context, the closed-formed expressions for the outage probability (OP) and intercept probability (IP) are derived corresponding to an amplify-and-forward (AF)-based relaying scheme, which is challenging and has not been studied before. Finally, the exactness of the mathematical analyses is validated through Monte Carlo simulations. Furthermore, the effects of various key parameters (e.g., channel estimation errors, satellite's transmit power, relay's transmit power, number of relays, and fading severity parameter) are examined. 
\end{abstract}

%\vspace*{-0.25cm}
\section{Introduction} \label{Introduction}
Recently, the Internet of Things (IoT) has been spread worldwide due to its various applications such as smart cities, smart farming, wearable devices, healthcare, and smart communications \cite{PhuBack2021, TanIoT2021, Tran2020FDUAV}. Nevertheless, the explosive growth of the number of IoT devices (IoTDs) has brought new challenges to the traditional cellular networks due to its restricted resources, e.g., fixed locations, limited spectrum, and power \cite{XianglingSat2020, van2021user, 9795695}. Furthermore, IoTDs are usually distributed in harsh environments such as forests, deserts, mountains, or maritime. Therefore, it is more difficult for the current cellular networks to provide seamless and cost-efficient services to a massive amount of IoTDs in these scenarios \cite{FangSat6G}. Thanks to recent projects such as SpaceX and OneWeb, SPUTNIX, thousands of low earth orbit (LEO) satellites are being launched to provide global coverage and high throughput supporting the traditional terrestrial communication networks \cite{OneWeb,SpaceX,DengUltraSat}. Therefore, satellite communications have become a promising solution to overcome the limitations of IoT networks. 

Besides many advantages, satellite communications (SatComs) are not without limitations \cite{van2021user,bui2021robust}. Due to considerable distance, obstacles, and shadowing effects between satellite and ground users (GUs), Line-of-sight (LoS) links between them are not always guaranteed \cite{vanelli2007satellite}. Therefore, the hybrid satellite-terrestrial relay network (HSTRN) has been proposed to leverage the benefits of both terrestrial and space communications to improve network performance \cite{KefengTVT2020,QingquanTAES2020,ZengTWSat2019,PanSatUAV2020,SharmaSatUAV2021,HieuSat2022}. In \cite{KefengTVT2020}, the authors investigated the HSTRN by considering the direct link between satellite and destination, interference at relay and destination, and the hardware impairment (HI) effects. In \cite{QingquanTAES2020}, the authors studied the HSTRN by optimizing the beam-forming factors at the relay to maximize the total capacity. Moreover, they proposed two scheduling schemes, namely user fairness and best user scheduling. The work in \cite{ZengTWSat2019} analyzed the performance analysis of two-way HSTRN, where the transmission links between relay to satellite and relay to users follow $\kappa$-$\mu$ shadowed fading and Nakagami-m distribution, respectively. In contrast to \cite{KefengTVT2020,QingquanTAES2020} only considered a fixed terrestrial user as a relay, an unmanned aerial vehicle (UAV) has been proposed to act as a relay in \cite{tran2019Asilomar,PanSatUAV2020,tran2020globecom,SharmaSatUAV2021,HieuSat2022}. 

Beyond the in-depth analysis of satellite-terrestrial networks in the works mentioned above, they did not consider the security aspect of their systems. Due to the broadcast nature and wide coverage range of SatComs, the security and privacy of this communication system are not guaranteed. Conventional cryptography methods have been applied at the upper layers to improve information security. Nevertheless, these encryption methods require a high computational burden due to service management and key distribution \cite{bankey2019physical,kalantari2015secrecy}. In this context, physical layer security (PLS) can be considered a complementary solution, where it can achieve secure transmission by leveraging the characteristics of wireless channels \cite{bankey2019physical,GouSatPLS2021,SharmaUAVSatPLS2020,XiaokaiSatNOMA2021,li2021physical,li2021hardware,li2020q}. The authors in \cite{bankey2019physical} studied the PLS in HSTRN in which a multi-antenna satellite tried to communicate with multi-antenna destinations via the help of multi-antenna relays in the presence of an eavesdropper. Non-orthogonal multiple access (NOMA) and HI were investigated with HSTRN in \cite{GouSatPLS2021}. Moreover, the authors in \cite{GouSatPLS2021} studied two security schemes, i.e., colluding and non-colluding scenarios. In \cite{SharmaUAVSatPLS2020}, UAVs were played as a flying relay (or an eavesdropper) to transfer (or overhear) information from source to destination in HSTRN. In  \cite{XiaokaiSatNOMA2021}, the authors proposed a novel hybrid satellite-aerial-terrestrial networks
(HSATNs), where ground users could obtain content from a cache-enabled UAV or a NOMA-based satellite. Moreover, the authors derived the outage and hit probability of the considered system model with stochastic geometry. Regarding the relay/backscatter-aided terrestrial communications, the authors in \cite{li2021physical,li2021hardware,li2020q} derived the closed-form expression of the secrecy channel capacity, OP, and IP to demonstrate the system robustness under attacks. 

The previous works have addressed the different new challenges in HSTRN, such as HI, NOMA, interference, two-way relaying, and UAV-enabled relaying. Nevertheless, none of these works consider a jammer to improve the information secrecy. Recently, few research works have investigated PLS in HSTRN with jammer \cite{GaofengTVT2020,HanSatTVT2021,MouniaSat2021,ai2019physical,SaiSat2021}. The work in \cite{GaofengTVT2020} studied the secrecy problem in multi-beam satellite systems by applying the cooperative jamming method. Moreover, they proposed an alternating algorithm by jointly optimizing the power allocation and beamforming factors. The authors in \cite{HanSatTVT2021} investigated a satellite and UAV system under hostile jamming environments. The authors designed the UAV trajectory to avoid jamming signals, where the UAV performed reconnaissance tasks and transferred the collected data to the satellite. The work in \cite{MouniaSat2021} investigated physical layer security in cognitive satellite-terrestrial networks with jamming, where dual-hop communications were considered under RF and optical links in the first and second hop, respectively. The authors in \cite{ai2019physical} investigated the closed-form expressions of the secrecy channel capacity and OP for a single relay and perfect CSI. In \cite{SaiSat2021}, the work considered a new system model in which intelligent reflecting surface (IRS)  enabled PLS in satellite-terrestrial networks. 

Despite significant achievements in the above studies, most of the previous works have concentrated on resource allocation perspectives under the availability of instantaneous channel state information (CSI) \cite{van2021user,bui2021robust, QingquanTAES2020, tran2019Asilomar,PanSatUAV2020,tran2020globecom,SharmaSatUAV2021,HieuSat2022, GaofengTVT2020, HanSatTVT2021}. In contrast, there are only a few works on deriving the analytical results for satellite-terrestrial integrated networks \cite{KefengTVT2020,bankey2019physical,GouSatPLS2021,MouniaSat2021}, but only based on perfect CSI and not for the IP and OP under the physical security perspectives with a  friendly jammer.  Consequently, there is still room for research on security and reliability in satellite-terrestrial (SatTer) networks. In this paper, we present an amplify-and-forward (AF)-based SatTer network for secure communication in the presence of an eavesdropper.\footnote{In this paper, we select the AF relaying for analysis since it is easy to implement in practice due to its simple hardware with the only main cost of an amplifier. In contrast, a DF has costly hardware with at least one radio frequency chain on other supplementary components to decode and separate the noise and desired signal. A consideration of the DF relaying is left for future work.} The main contributions of this work are given as follows:

\begin{itemize}
	\item To the best of our knowledge, this is the first work that mathematically obtains the closed-form expressions for the IP and OP in AF-based SatTer networks in the presence of an eavesdropper with imperfect CSI and generalized shadowed Rician channels. In particular, a friendly jammer is considered to broadcast the artificial noise to decrease the eavesdropping ability of the illegitimate party. Furthermore, the closed-form expressions based on particular CSI offer low-cost designs to evaluate the IP and OP in practice.
	\item Monte Carlo simulations are performed to verify the correctness of mathematical analyses for various parameter settings. Moreover, the simulations also provide the trade-off between security and reliability in the considered system model.
	\item Insightful discussions on the impact of different key parameters of SatTer networks on its security-reliability trade-off are numerically provided. Specifically, the satellite's transmit power, number of relays, and relay transmit power can be suitably selected to reduce the eavesdropper's effects and enhance the system performance.
\end{itemize}

The remainder of the paper is organized as follows: The system model and problem formulation are given in Section~\ref{System_Model}. The derivation
of key performance metrics, including the OP and throughput of the proposed model, is presented in Section~\ref{sec:3}. Numerical results are shown in Section~\ref{sec:4}. Section~\ref{sec:5} concludes the paper.

\section{System Model} \label{System_Model}
We study the downlink (DL) scenario of the physical layer security in an HSTRN. The system model consists of one satellite S communicating with one destination D via the support of multiple relays ${\rm R}_n$ with $n \in \{1,\dots, N\}$. Moreover, due to the masking effect and severe shadowing, it does not exist a direct link between $\rm S \to \rm D$ \cite{BhatnagarSat2013,an2015performance,bankey2019physical}. Besides, an eavesdropper E tries to wiretap confidential information from relay $\rm{R}_n$ \cite{ha2020security,phan2021study,tran2019secrecy}. To reduce the eavesdropping ability of the illegitimate party, we deploy one friendly jammer J that broadcasts the artificial noise.\footnote{
Broadcasting the artificial noise is one solution for physical security as no prior information about the network is available. Accordingly, the energy consumption for a friendly jammer is necessary to boost the security of the legitimate destination. Some particular devices can be considered as friendly jammers to broadcast artificial noise, for instance, in military applications \cite{sharma2020artificial}. The legitimate destination knows the signals transmitted from the friendly jammer using a specific signal encoding structure \cite{gu2019secrecy}.} Notably, the jamming noise is known to the legitimate destination and it can be canceled, while it is unknown to the eavesdropper \cite{zou2016physical, MouniaSat2021}. 
\begin{figure}[t]
	\label{Fig1}
	\centering
	\includegraphics[height=5.2cm,width = 8cm]{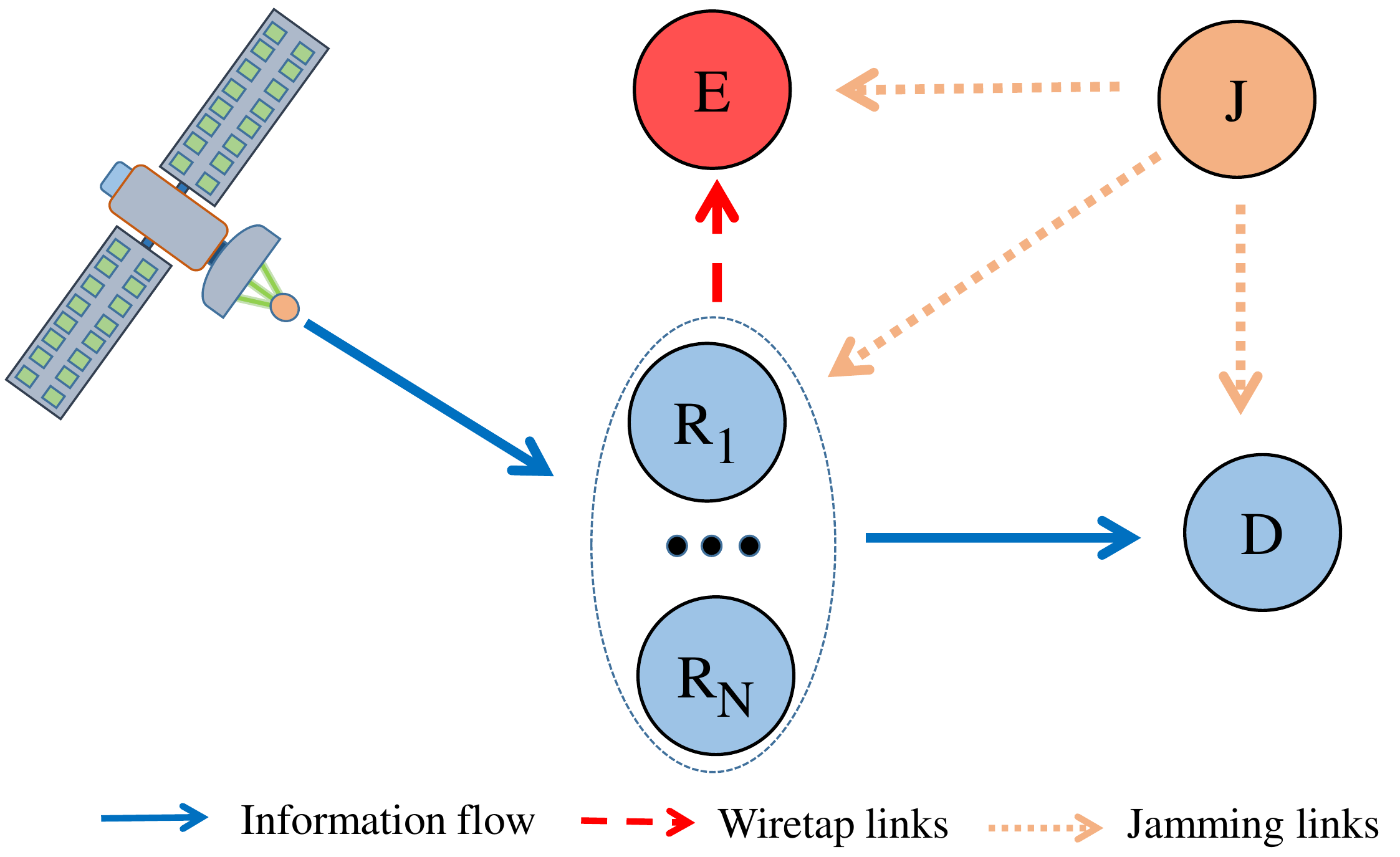}
	\caption{The considered system model}
\end{figure}
\subsection{Channel model}
In practice, it is difficult to obtain perfect channel state information (CSI) due to channel estimation errors (CEEs). Consequently, channel estimation methods, such as linear minimum mean square error (MMSE) channel estimator \cite{li2019security} are necessary to obtain CSI. Therefore, the channel can be modeled as 
\begin{align}
	\label{EQ1}
	{\mathord{\buildrel{\lower3pt\hbox{$\scriptscriptstyle\frown$}} 
\over h} _i} = {h_i} + {e_i},
\end{align}
where ${e_i}$ with $i \in \left( {{\rm{S}}{{\rm{R}}_n}{\rm{,}}{{\rm{R}}_n}{\rm{D}},{{\rm{R}}_n}{\rm{E}},{\rm{JE}},{\rm{J}}{{\rm{R}}_n}{\rm{,JD}}} \right)$ is the CEEs with ${e_i} \sim CN\left( {0,\mu _i^2} \right)$ and ${\mathord{\buildrel{\lower3pt\hbox{$\scriptscriptstyle\frown$}} \over h} _i}$ signifies the estimated channel of the real channel ${h_i}$. Without loss of generality, let us assume that the estimated channels between two GUs in terrestrial communications follow block Rayleigh fading, whereas channel coefficients are unchanged during one transmission block and vary over different blocks.\footnote{The channel model considered in this paper matches well with actual measurements from the propagation environments with rich scatterers around the receivers and without a dominant path \cite{9531522}. A more general model such as \cite{yacoub2007alpha} should be considered in a future work.} As a result, the channel gains can be represented as ${\gamma _{{{\rm{R}}_n}{\rm{E}}}} = {\left| {{h_{{{\rm{R}}_n}{\rm{E}}}}} \right|^2},{\gamma _{{{\rm{R}}_n}{\rm{D}}}} = {\left| {{h_{{{\rm{R}}_n}{\rm{D}}}}} \right|^2}$, etc. are exponential random variables (RVs). More specifically, the probability density function (PDF) and cumulative distribution function (CDF) are respectively given as \cite{nguyen2019performance,tin2021performance}:
\begin{align}
	\label{EQ2}
	{F_X}(x) = 1 - \exp ( - \lambda x),
\end{align}
\begin{align}
	\label{EQ3}
	{f_X}(x) = \frac{{\partial {F_X}(x)}}{{\partial x}} = \lambda \exp ( - \lambda x),
\end{align}
where $\lambda$ is the rate parameter of $X$. Next, we consider Shadowed-Rician fading model for the satellite links \cite{an2015performance}, which has been widely applied in different satellite services using the S-band, UHF-band, Ka-band, and L-band \cite{an2016secure}. 
Therein, the PDF of ${\gamma _{{\rm{S}}{{\rm{R}}_n}}} = {\left| {{h_{{\rm{S}}{{\rm{R}}_n}}}} \right|^2}$ between ${\rm S} \to {\rm R}_n$ is given by \cite{BhatnagarSat2013,bankey2019physical}
\begin{align}
	\label{EQ4}
	{f_{{\gamma _{{\rm{S}}{{\rm{R}}_n}}}}}(x) = {\alpha _n}\exp \left( { - {\beta _n}x} \right){}_1{F_1}\left( {{m_n};1;{\delta _n}x} \right),  x \ge 0,
\end{align}
where ${\alpha _n} \triangleq {\beta _n}{\left( {\frac{{2{b_n}{m_n}}}{{2{b_n}{m_n} + {\Omega _n}}}} \right)^{{m_n}}},$ ${\beta _n} \triangleq  \frac{1}{{2{b_n}}},$ ${\delta _n} \triangleq \frac{{{\beta _n}{\Omega _n}}}{{2{b_n}{m_n} + {\Omega _n}}}$, with ${\Omega _n}$ and $2{b_n}$ denote the average power of LoS and multi-path components at relay ${\rm R}_n$, respectively. Moreover, ${m_n}$ implies the fading severity parameter at relay ${\rm R}_n$ and ${}_1{F_1}\left( {.;.;.} \right)$ denotes the confluent hyper-geometric function of the first kind \cite[Eq.9.210.1]{jeffrey2007table}. For arbitrary integer-valued fading severity parameters, one can simplify ${}_1{F_1}\left( {{m_n};1;{\delta _n}x} \right)$ in \eqref{EQ3} to express the PDF of ${\gamma _{{\rm{S}}{{\rm{R}}_n}}}$ as \cite{upadhyay2015max}
\begin{align}
	\label{EQ5}
	{f_{{\gamma _{{\rm{S}}{{\rm{R}}_n}}}}}(x) = {\alpha _n}\sum\limits_{k = 0}^{{m_n} - 1} {{\zeta _n}(k)} {x^k}\exp \left( { - \left( {{\beta _n} - {\delta _n}} \right)x} \right),
\end{align}
where the following definition holds
\begin{equation} \label{eq:zetan}
{\zeta _n}(k) = \frac{{{{( - 1)}^k}{{\left( {1 - {m_n}} \right)}_k}\delta _n^k}}{{{{\left( {k!} \right)}^2}}},
\end{equation}
 and ${\left(  .\right)_k}$ is the Pochhammer symbol \cite[Eq. p. xliii]{jeffrey2007table}. The corresponding CDF denoted by ${F_{{\gamma _{{\rm{S}}{{\rm{R}}_n}}}}}(x)$ is formulated as
 \begin{equation}
 {F_{{\gamma _{{\rm{S}}{{\rm{R}}_n}}}}}(x) = \Pr \left( {{\gamma _{{\rm{S}}{{\rm{R}}_n}}} < x} \right) = \int\limits_0^x {{f_{{\gamma _{{\rm{S}}{{\rm{R}}_n}}}}}(x)dx},
 \end{equation}
which can be obtained in closed form as follows:
\begin{align}
	\label{EQ6}
&{F_{{\gamma _{{\rm{S}}{{\rm{R}}_n}}}}}(x) = 1 - {\alpha _n}\sum\limits_{k = 0}^{{m_n} - 1} {\sum\limits_{p = 0}^k {\frac{{{\zeta _n}(k)k!}}{{p!}}} } {\vartheta_1^{ - A_1 }}  {x^p}\exp \left( { - \vartheta_1 x} \right) \notag\\
& = 1 - {\alpha _n}\sum\limits_{k = 0}^{{m_n} - 1} \sum\limits_{p = 0}^k {\frac{{{{( - 1)}^k}{{\left( {1 - {m_n}} \right)}_k}\delta _n^k}}{{k!p!}}} { \vartheta_1^{ - A_1}} {x^p}\exp \left( { - \vartheta_1 x} \right),
\end{align}
where $A_1 \triangleq k+1-p,$ $\vartheta_1 \triangleq \beta_n - \delta_n$. The first equality in \eqref{EQ6} is obtained by applying  \cite[Eq.3.351.1]{jeffrey2007table}, while the second equality in \eqref{EQ6} is achieved by utilizing the detail expression of $\zeta_n (k)$ defined in \eqref{eq:zetan} and then doing some algebra.  We hereafter exploit the fundamental achievement in \eqref{EQ6} to work on the analysis of the IP and OP.
\subsection{Information transmission}
In this subsection, we describe the AF relaying protocol that is used to forward the data from ${\rm S} \to {\rm R}_n$. Because all GUs work in the half-duplex mode, the signal transmission from ${\rm S} \to {\rm D}$ through two-time phases \cite{Tan2020Compnet,nguyen2020outage}. In the first one, S transmits its data ${x_{\rm{S}}}$ to the relay ${\rm{R}}_n$. Therefore, the received signal at ${\rm{R}}_n$ can be expressed as:
\begin{align}
	\label{EQ7}
	{y_{{{\rm{R}}_n}}} = {\mathord{\buildrel{\lower3pt\hbox{$\scriptscriptstyle\frown$}} 
\over h} _{{\rm{S}}{{\rm{R}}_n}}}{x_{\rm{S}}} + {n_{{{\rm{R}}_n}}} = \left( {{h_{{\rm{S}}{{\rm{R}}_n}}} + {e_{{\rm{S}}{{\rm{R}}_n}}}} \right){x_{\rm{S}}} + {n_{{{\rm{R}}_n}}},
\end{align}
where ${x_{\rm{S}}}$ is the symbol signal with ${\rm E}\{ {{{| {{x_{\rm{S}}}} |}^2}} \} = {P_{\rm{S}}}$, ${\rm E}\left\{  \cdot  \right\}$ denotes the expectation operation; ${n_{{{\rm{R}}_n}}}$ is the zero mean additive white Gaussian noise (AWGN) with variance ${N}_0$ at relay ${\rm R}_n$. 

In the second phase, relay ${\rm R}_n$ first amplifies received signal ${y_{{{\rm{R}}_n}}}$ with scale parameter G and then relays it to the destination D. At the same time, the eavesdropper can detect the transmitted signal from ${\rm R}_n$ and try to overhear the confidential information. The cooperative jamming technique can be used to reduce the eavesdropping link's quality. Specifically, the single antenna friendly jammer (J) is employed to continuously generate artificial noises to the eavesdropper. Furthermore, the relays and destination are the legitimate users and are assumed to be known the coding sequence of the jammer. Therefore, they can cancel the jammer's interference. Furthermore, it is assumed that the eavesdropper cannot intercepts signal from $\rm S$ due to the masking effect or severe fading. Therefore, the  eavesdropper monitors only transmitted signal from relay ${\rm R}_n$ \cite{huang2018secrecy,bankey2019physical}. Consequently, the received signal at the destination $\rm D$ and the eavesdropper $\rm E$  are respectively given as
\begin{align}
	\label{EQ8}
{y_{\rm{D}}} &= \left( {{h_{{{\rm{R}}_n}{\rm{D}}}} + {e_{{{\rm{R}}_n}{\rm{D}}}}} \right){\rm{G}}{{\rm{y}}_{{{\rm{R}}_n}}} + {n_{\rm{D}}}\notag\\
& = \big( {{h_{{{\rm{R}}_n}{\rm{D}}}} + {e_{{{\rm{R}}_n}{\rm{D}}}}} \big){\rm{G}}\big[ \big( {{h_{{\rm{S}}{{\rm{R}}_n}}} + {e_{{\rm{S}}{{\rm{R}}_n}}}} \big){x_{\rm{S}}}   +{n_{{{\rm{R}}_n}}} \big] + {n_{\rm{D}}},\\
	\label{EQ9}
{y_{\rm{E}}} &= \left( {{h_{{{\rm{R}}_n}{\rm{E}}}} + {e_{{{\rm{R}}_n}{\rm{E}}}}} \right){\rm{G}}\big[ \left( {{h_{{\rm{S}}{{\rm{R}}_n}}} + {e_{{\rm{S}}{{\rm{R}}_n}}}} \right){x_{\rm{S}}} \notag\\ & + {n_{{{\rm{R}}_n}}} \big]  + \big( {{h_{{\rm{JE}}}} + {e_{{\rm{JE}}}}} \big){x_{\rm{J}}} + {n_{\rm{E}}},
\end{align}
where ${n_{\rm{D}}}$ and ${n_{\rm{E}}}$ are the zero mean AWGN with variance $N_0$ at ${n_{\rm{D}}}$ and ${n_{\rm{E}}}$, respectively; ${x_{\rm{J}}}$ the transmit signal at jammer $\rm{J}$ and ${\rm E}\left\{ {{{\left| {{x_{\rm{J}}}} \right|}^2}} \right\} = {P_{\rm{J}}}$. Based on \eqref{EQ7}, the scale parameter G can be determined as
\begin{align}
	\label{EQ10}
	{\rm{G = }}\sqrt {\frac{{{P_{\rm{R}_n}}}}{{{P_{\rm{S}}}\left( {{\gamma _{{\rm{S}}{{\rm{R}}_n}}} + \mu _{{\rm{S}}{{\rm{R}}_n}}^2} \right) + {N_0}}}}.
\end{align}
where $P_{\mathrm{R}_n}$ is  the transmit power of $\mathrm{R}_n$. For simplicity, we assume that CEEs values at each S-${\rm{R}}_n$ or ${\rm{R}}_n$-D or ${\rm{R}}_n$-E links are the same, i.e., $\mu _{{\rm{S}}{{\rm{R}}_1}}^2 = \mu _{{\rm{S}}{{\rm{R}}_2}}^2 = \dots = \mu _{{\rm{S}}{{\rm{R}}_n}}^2 = \mu _{{\rm{SR}}}^2,\forall n \in \left( {1,2,\dots,N} \right)$, etc. From \eqref{EQ8}, \eqref{EQ9} and \eqref{EQ10}, the signal to noise ratio (SNR) at D and eavesdropper ${\rm E}$ can be given as
\begin{align}
	\label{EQ11}
	{\gamma _{\rm{D}}} &= \frac{{\Psi \Phi {\gamma _{{\rm{S}}{{\rm{R}}_n}}}{\gamma _{{{\rm{R}}_n}{\rm{D}}}}}}{{{\gamma _{{{\rm{R}}_n}{\rm{D}}}}\Phi \vartheta_2 + {\gamma _{{\rm{S}}{{\rm{R}}_n}}}\Psi \left( {\mu _{{\rm{RD}}}^2\Phi  + 1} \right) + \Xi }},\\
	\label{EQ12}
	{\gamma _{\rm{E}}} &= \frac{{\Psi \Phi {\gamma _{{\rm{S}}{{\rm{R}}_n}}}{\gamma _{{{\rm{R}}_n}{\rm{E}}}}}}{{{\gamma _{{{\rm{R}}_n}{\rm{E}}}}\Phi \vartheta_2 + {\gamma _{{\rm{S}}{{\rm{R}}_n}}}\Psi \Lambda_1 + {\gamma _{{\rm{JE}}}}\Theta \vartheta_2  + \Lambda }},
\end{align}
where $\Psi \triangleq \frac{{{P_{\rm{S}}}}}{{{N_0}}},$ $\Xi \triangleq \left( {\Psi \mu _{{\rm{SR}}}^2 + 1} \right)\left( {\Phi \mu _{{\rm{RD}}}^2 + 1} \right),$ $\Phi \triangleq \frac{{{P_{\rm{R}}}}}{{{N_0}}},$ $	\vartheta_2 \triangleq \left( {\Psi \mu _{{\rm{SR}}}^2 + 1} \right),$ $\Theta  \triangleq \frac{{{P_{\rm{J}}}}}{{{N_0}}},$ $\Lambda  \triangleq \left( {\Psi \mu _{{\rm{SR}}}^2 + 1} \right)\left( {\Phi \mu _{{\rm{RE}}}^2 + \Theta \mu _{{\rm{JE}}}^2 + 1} \right),$ $\Lambda_1  \triangleq \left( {\mu _{{\rm{RE}}}^2\Phi  + \Theta \mu _{{\rm{JE}}}^2 + \Theta {\gamma _{{\rm{JE}}}} + 1} \right).$ In \eqref{EQ11}, we assume that the legitimate destination knows the coding sequence of the friendly jammer, they can exploit the same mechanism to generate artificial noise , e.g., using a standard hardware with the same random seed \cite{zou2016physical,MouniaSat2021} or the orthogonality between artificial noise and legitimate channel \cite{gu2019secrecy}.
\begin{remark}
In this work, we apply the partial relay selection (PRS) method. Specifically, we propose a sub-optimal relay selection protocol in which the best relay $a$-th is selected, as follows:\footnote{The partial relay selection in \eqref{EQ13} is sufficiently simple to implement in practice where the requirements of low latency are needed. However, it is a sub-optimal selection, and a better one is left for future work.}
\begin{align}
	\label{EQ13}
	a &= \underbrace {\argmax }_{n = 1,2,\dots,N}\left\{ {{\gamma _{{\rm{S}}{{\rm{R}}_n}}}} \right\}  \Leftrightarrow {\gamma _{{\rm{S}}{{\rm{R}}_a}}} = \underbrace {\max }_{n = 1,2,\dots,N}\left\{ {{\gamma _{{\rm{S}}{{\rm{R}}_n}}}} \right\}.
\end{align}
\end{remark}
To the best of our knowledge, all mentioned works in the literature, for example,  \cite{bankey2019physical,upadhyay2015max}, only considered identical independent distribution (i.i.d.) Shadowed-Rician fading channels with the same parameters, i.e., $m_n$, $\Omega_n$, and $b_n$, between ${\rm S} \to {\rm R}_n, \forall n = \{1,\dots,N\}$. Motivated by these observations, in this work, we investigate the generalized i.i.d. Shadowed-Rician channel, where S can transmit its information to each relay with different parameters. Alternatively, our framework is a generic version of the previous works. The CDF of ${\gamma _{{\rm{S}}{{\rm{R}}_a}}}$ can be computed as:
\begin{align}
	\label{EQ14}
	{F_{{\gamma _{{\rm{S}}{{\rm{R}}_a}}}}}(x) &= \Pr \left( {\underbrace {\max }_{n = 1,2,...,N}\left\{ {{\gamma _{{\rm{S}}{{\rm{R}}_n}}}} \right\} < x} \right) = \prod\limits_{n = 1}^N {{F_{{\gamma _{{\rm{S}}{{\rm{R}}_n}}}}}(x)}.
\end{align}
In order for further processing, by substituting \eqref{EQ6} into \eqref{EQ14}, it yields the following equality
%\begin{figure*}[h!]
\begin{align}
	\label{EQ15}
	{F_{{\gamma _{{\rm{S}}{{\rm{R}}_a}}}}}(x) =& \prod\limits_{n = 1}^N \Bigg( 1 - {\alpha _n}\sum\limits_{k = 0}^{{m_n} - 1} {\sum\limits_{p = 0}^k {\frac{{{{( - 1)}^k}{{\left( {1 - {m_n}} \right)}_k}{{\left( {{\delta _n}} \right)}^k}}}{{k!p!}}} } \notag\\ &\vartheta_1^{ - \left( {k + 1 - p} \right)} \times {x^p}\exp \big( { - \vartheta_1 x} \big) \Bigg).
\end{align}
%\setlength{\arraycolsep}{1pt}
%\hrulefill \setlength{\arraycolsep}{0.0em}
%\vspace*{1pt}
%\end{figure*}
Let us denote  the second term of the product in \eqref{EQ15} as
${x_n} = \sum\limits_{k = 0}^{{m_n} - 1} \sum\limits_{p = 0}^k \frac{{{\alpha _n}{{( - 1)}^k}{{\left( {1 - {m_n}} \right)}_k}{{\left( {{\delta _n}} \right)}^k}}}{{k!p!{\vartheta_1^{A_1}}}}  {x^p}\exp \left( { - \vartheta_1 x} \right)$, then we obtain the following result
\begin{align}
	\label{EQ16}
{F_{{\gamma _{{\rm{S}}{{\rm{R}}_a}}}}}(x) &= \prod\limits_{n = 1}^N {\left( {1 - {x_n}} \right)} 
\notag\\ & = 1 + \sum\limits_{n = 1}^N {\frac{{{{( - 1)}^n}}}{{n!}}} \underbrace {\sum\limits_{{n_1} = 1}^N \dots \sum\limits_{{n_n} = 1}^N {} }_{{n_1} \ne {n_2} \dots \ne {n_n}}\prod\limits_{t = 1}^n {{x_{{n_t}}}}.
\end{align}
We note that the result in \eqref{EQ16} is obtained from \eqref{EQ15} by  induction. From the obtained result in \eqref{EQ16}, we further manipulate the expression ${x_{{n_t}}}$ follows:
\begin{align}
	\label{EQ17}
	&\prod\limits_{t = 1}^n {{x_{{n_t}}}}  = \sum\limits_{{k_1} = 0}^{{m_{{n_t}}} - 1} \sum\limits_{{p_1} = 0}^{{k_1}} \dots\sum\limits_{{k_n} = 0}^{{m_{{n_t}}} - 1} \sum\limits_{{p_n} = 0}^{{k_n}} \prod\limits_{t = 1}^n {\alpha _{{n_t}}} \times\notag\\
&\frac{{{{( - 1)}^{{k_t}}}{{\left( {1 - {m_{{n_t}}}} \right)}_{{k_t}}}{{\left( {{\delta _{{n_t}}}} \right)}^{{k_t}}}}}{{{k_t}!{p_t}!{{\vartheta_3}^{A_3}}}}   {x^{A_2 }}\exp \left( { - \sum\limits_{t = 1}^n {\vartheta_3 x} } \right)    ,
\end{align}
where $A_2 \triangleq \sum\limits_{t = 1}^n {{p_t}}, $
$A_3 \triangleq k_t + 1 - p_t,$ $\vartheta_3 \triangleq  {{\beta _{{n_t}}} - {\delta _{{n_t}}}}$. After that, the closed-form expression of the  CDF of ${\gamma _{{\rm{S}}{{\rm{R}}_a}}}$ is given in Lemma~\ref{lemmaF}.
\begin{lemma}\label{lemmaF}
	Based on \eqref{EQ16} and \eqref{EQ17}, ${F_{{\gamma _{{\rm{S}}{{\rm{R}}_a}}}}}(x)$ can be expressed as
	\begin{align}
		\label{EQ18}
		&{F_{{\gamma _{{\rm{S}}{{\rm{R}}_a}}}}}(x) = 1 + \sum\limits_{n = 1}^N \frac{{{{( - 1)}^n}}}{{n!}}\sum\limits_{}^* {\alpha _{{n_t}}} \notag\\&\times \frac{{{{( - 1)}^{{k_t}}}{{\left( {1 - {m_{{n_t}}}} \right)}_{{k_t}}}{{\left( {{\delta _{{n_t}}}} \right)}^{{k_t}}}}}{{{k_t}!{p_t}!{\vartheta_3^{A_3}}}}  {x^{A_2}}\exp \left( { - \sum\limits_{t = 1}^n {\vartheta_3 x} } \right) ,
	\end{align}
	where the summation notation $\sum\limits_{}^*$ is defined as
	\begin{align}
		\sum\limits_{}^*  =  \underbrace {\sum\limits_{{n_1} = 1}^N \dots\sum\limits_{{n_n} = 1}^N {} }_{{n_1} \ne {n_2}\dots \ne {n_n}}\sum\limits_{{k_1} = 0}^{{m_{{n_t}}} - 1} {\sum\limits_{{p_1} = 0}^{{k_1}} {\dots\sum\limits_{{k_n} = 0}^{{m_{{n_t}}} - 1} {\sum\limits_{{p_n} = 0}^{{k_n}} {\prod\limits_{t = 1}^n . } } } }
	\end{align}
\end{lemma}
The preliminary result on the CDF of ${\gamma _{{\rm{S}}{{\rm{R}}_a}}}$ visualizes the influences of all the available relays to the partial relay selection and the SNR value. We hereafter utilize the key result in \eqref{EQ18} to analyze the IP and OP.
\vspace*{-0.2cm}
\section{Performance Analysis}
This section derives the analytical results of the OP and IP with insightful observations.
\vspace*{-0.2cm}
\subsection{Outage Probability (OP) Analysis}
\label{sec:3}
For a given maximum SNR value, the OP of system can be defined as:
\begin{align} 
		\label{EQ19}
		{\rm{OP}} = \Pr \left( {{\gamma _{\rm{D}}} < {\gamma _{th}}} \right),
	\end{align} 	
where ${\gamma _{th}} = {2^{2{C_{th}}}} - 1$, ${C_{th}}$ is the predefined threshold of the system, and $\Pr(\cdot)$ is the probability of an event.  
\begin{theorem}
	\label{theorem_1}
	The analytical expression of the OP can be mathematically represented as
	\begin{align}
		\label{EQ24}
		{\rm{OP}} &= 1 + \sum\limits_{n = 1}^N \frac{{{{( - 1)}^n}}}{{n!}}\notag\\ & \sum\limits_{}^ *  {{\alpha _{{n_t}}}\frac{{2{{( - 1)}^{{k_t}}}{{\left( {1 - {m_{{n_t}}}} \right)}_{{k_t}}}{{\left( {{\delta _{{n_t}}}} \right)}^{{k_t}}}{\lambda _{{\rm{RD}}}}}}{{{k_t}!{p_t}!{\vartheta_3^{A_3}}}}}  \notag \\
		& \times {\left( {\frac{{{\gamma _{th}}}}{\Psi }} \right)^{A_2 }}\exp \left( { - \sum\limits_{t = 1}^n {\vartheta_3\frac{{{\gamma _{th}}\vartheta_2}}{\Psi }}  - {\lambda _{{\rm{RD}}}}{{\tilde \gamma }_{th}}} \right) %\notag \\
	\end{align}
\begin{align*}
		& \times \sum\limits_{q = 0}^{A_2} {\left( \begin{array}{l}
				A_2 \\
				q
			\end{array} \right)} {\vartheta_2^{A_2 - q}}{{\tilde \Xi }^q}{\left( {\frac{{\sum\nolimits_{t = 1}^n {\vartheta_3{\gamma _{th}}\tilde \Xi } }}{{{\lambda _{{\rm{RD}}}}\Psi }}} \right)^{\frac{{ - q + 1}}{2}}} \notag\\ &{K_{ - q + 1}}\left( {2 \sqrt {\sum\limits_{t = 1}^n {\vartheta_3 \frac{{{\gamma _{th}}{\lambda _{{\rm{RD}}}}\tilde \Xi }}{\Psi }} } } \right),
	\end{align*}
	where ${K_v}\left(  \cdot  \right)$ is the modified Bessel function of the second kind and v-th order.
\end{theorem}
\begin{proof}
See Appendix~\ref{Appendix:theorem_1}.
\end{proof}
From \eqref{EQ24}, the OP is upper bounded by one ideally. Nonetheless, the complicated term shows a reduction in the amount of the OP is complicatedly expressed by several aspects, for example, the relay selection and the effectiveness of the friendly jammer. We notice that an  asymptotic analysis is challenging and left for future work since the transmit power of the satellite and the selected relay is independent of each other. 
%{I have UPDATED THE SIMULATION TABLE}
{\begin{table*}[t]
		\caption{Simulation parameters.}
		\centering
		\label{Table_2}
		\setlength{\tabcolsep}{3pt}
		%|l|l|l|l|
		\begin{tabular}{|>{\centering\arraybackslash} m{2cm} | >{\centering\arraybackslash} m{7cm}| >{\centering\arraybackslash} m{3cm}| >{\centering\arraybackslash} m{4cm}|}
			\hline 
			\textbf{Symbol} & \textbf{Parameter name} & \textbf{Fixed value} & \textbf{Varying range} \\ \hline \hline
			$C_{\rm th}$&   SNR threshold of the system&     1 & none\\
			$\mu_{\rm{SR}}$& CEEs of S-R link& 0.25 & none\\
			$\mu_{\rm{RD}}$& CEEs of R-D link& 0.25 & none\\
			$\mu_{\rm{RE}}$& CEEs of R-E link& 0.25 & none\\
			$\mu_{\rm{JE}}$& CEEs of R-E link& 0.25 & none\\
			$\Psi$& Transmit power-to-noise-ratio from sattelite & 20 dB& 0 to 50 dB \\
			$\Phi$& Transmit power-to-noise-ratio from relay & 10 dB& none \\
			$\Theta$& Transmit power-to-noise-ratio from jammer & 1 dB& none \\
			$N$ & Number of relays & 1-3 & none \\
			\hline \hline
		\end{tabular}
		\label{tab1}
		\vspace{-0.4cm}
\end{table*}}
\vspace*{-0.2cm}
\subsection{Intercept Probability (IP) Analysis}
The confidential information can be intercepted if eavesdropper successfully decodes received signals, i.e., ${\gamma _{\rm{E}}} \ge {\gamma _{th}}$. Therefore, the IP can be defined by following the similar methodology in \cite{li2020secrecy,YulongPLS2015} as follows
\begin{align}
	\label{EQ25}
	{\rm{IP}} = \Pr \left( {{\gamma _{\rm{E}}} \ge {\gamma _{th}}} \right) = 1 - \Pr \left( {{\gamma _{\rm{E}}} < {\gamma _{th}}} \right).
\end{align}
In order to obtain the analytical from, we base on \eqref{EQ12} and \eqref{EQ25} to reformulate the IP  as follows
\begin{align}
	\label{EQ26}
&{\rm{IP}} = 1 - \Pr \left( {\frac{{\Psi \Phi {\gamma _{{\rm{S}}{{\rm{R}}_a}}}{\gamma _{{{\rm{R}}_a}{\rm{E}}}}}}{\Lambda_6} < {\gamma _{th}}} \right) \notag\\
& = 1 - \int_0^{ + \infty } Q \times {f_{{\gamma_{{\rm{JE}}}}}}(x)dx ,
\end{align}
where $Q \triangleq \Pr \left( {\frac{{\Psi \Phi {\gamma _{{\rm{S}}{{\rm{R}}_a}}}{\gamma _{{{\rm{R}}_a}{\rm{E}}}}}}{\Lambda_6} < {\gamma _{th}}} \right),$ $\Lambda_6 \triangleq
	{{\gamma _{{{\rm{R}}_a}{\rm{E}}}}\Phi \vartheta_2 + {\gamma _{{\rm{S}}{{\rm{R}}_a}}}\Psi \vartheta_5 + {\gamma _{{\rm{JE}}}}\Theta \vartheta_2 + \Lambda }$, and $\vartheta_5 \triangleq \left( {\mu _{{\rm{RE}}}^2\Phi  + \Theta \mu _{{\rm{JE}}}^2 + \Theta {\gamma _{{\rm{JE}}}} + 1} \right)$. From on \eqref{EQ26}, $Q$ can be calculated as
\begin{align}
	\label{EQ27}
Q &= \Pr \Big( {\Psi \Phi {\gamma _{{\rm{S}}{{\rm{R}}_a}}}{\gamma _{{{\rm{R}}_a}{\rm{E}}}} < {\gamma _{th}} \big(\Lambda_7+ \Lambda_8\big)} \Big)\notag\\
& = \Pr \left( {\Psi {\gamma _{{\rm{S}}{{\rm{R}}_a}}}\left[ {\Phi {\gamma _{{{\rm{R}}_a}{\rm{E}}}} - {\gamma _{th}}\Lambda_8 } \right] < {\gamma _{th}}\Lambda_7} \right)\notag\\
& = \left\{ \begin{array}{l}
\Pr \left( {\gamma _{{\rm{S}}{{\rm{R}}_a}}} < \frac{{{\gamma _{th}}\Lambda_7}}{\Lambda_9} \right),{\rm{if}}\,{\gamma _{{{\rm{R}}_a}{\rm{E}}}} > \frac{{{\gamma _{th}} \vartheta_5 }}{\Phi } \\
1\hfill,{\rm{if}}\,{\gamma _{{{\rm{R}}_a}{\rm{E}}}} \le \frac{{{\gamma _{th}}\vartheta_5}}{\Phi }
\end{array} \right. \notag\\
& = \int_0^{\Delta (x)} {{f_{{\gamma _{{{\rm{R}}_a}{\rm{E}}}}}}(y)dy} \notag\\ & + \int_{\Delta (x)}^{ + \infty } {{F_{{\gamma _{{\rm{S}}{{\rm{R}}_a}}}}}\left( {\frac{{{\gamma _{th}}\left[ {y\vartheta_2 + \Omega (x)} \right]}}{{\Psi \left[ {y - \Delta (x)} \right]}}} \right)  {f_{{\gamma _{{{\rm{R}}_a}{\rm{E}}}}}}(y)dy}.
\end{align}
where $\Lambda_7 \triangleq \left[ {{\gamma _{{{\rm{R}}_a}{\rm{E}}}}\Phi \vartheta_2 + x\Theta \vartheta_2 + \Lambda } \right]$, $\Lambda_8 \triangleq {\gamma _{{\rm{S}}{{\rm{R}}_a}}}\Psi \big( {\mu _{{\rm{RE}}}^2\Phi  + \Theta \mu _{{\rm{JE}}}^2 + \Theta x + 1} \big),$ $\Lambda_9 \triangleq \Psi \left[ {\Phi {\gamma _{{{\rm{R}}_a}{\rm{E}}}} - {\gamma _{th}}\vartheta_5} \right]$, $\Omega (x) \triangleq \frac{{x\Theta \vartheta_2 + \Lambda }}{\Phi },$ $\Delta (x) \triangleq \frac{{{\gamma _{th}}\vartheta_5s}}{\Phi } = {\mathord{\buildrel{\lower3pt\hbox{$\scriptscriptstyle\frown$}}\over \gamma } _{th}} + \frac{{{\gamma _{th}}\Theta x}}{\Phi }, $
${\mathord{\buildrel{\lower3pt\hbox{$\scriptscriptstyle\frown$}} \over \gamma } _{th}} \triangleq \frac{{{\gamma _{th}}\left( {\mu _{{\rm{RE}}}^2\Phi  + \Theta \mu _{{\rm{JE}}}^2 + 1} \right)}}{\Phi }.$ By using \eqref{EQ18},  $Q$ in \eqref{EQ27} can be thus re-written as 
\begin{align}
	\label{EQ28}
Q &= 1 + \sum\limits_{n = 1}^N \frac{{{{( - 1)}^n}}}{{n!}} \notag\\ &\sum\limits_{}^ *  {{\alpha _{{n_t}}}\frac{{{{( - 1)}^{{k_t}}}{{\left( {1 - {m_{{n_t}}}} \right)}_{{k_t}}}{{\left( {{\delta _{{n_t}}}} \right)}^{{k_t}}}{\lambda _{{\rm{RE}}}}}}{{{k_t}!{p_t}!{\vartheta_3^{A_3}}}}}  \notag\\
& \times \int_{\Delta (x)}^{ + \infty } {{\left( {\frac{{{\gamma _{th}}\left[ {y\vartheta_2 + \Omega (x)} \right]}}{{\Psi \left[ {y - \Delta (x)} \right]}}} \right)}^{\sum\limits_{t = 1}^n {{p_t}} }} \notag\\ &\exp \left( { - \sum\limits_{t = 1}^n {\vartheta_3\left( {\frac{{{\gamma _{th}}\left[ {y\vartheta_2 + \Omega (x)} \right]}}{{\Psi \left[ {y - \Delta (x)} \right]}}} \right)}  - {\lambda _{{\rm{RE}}}}y} \right) dy,
\end{align}
Let us denote $z = y - \Delta (x)$, then \eqref{EQ28} can be expressed as 
\begin{align}
	\label{EQ29}
&Q = 1 + \sum\limits_{n = 1}^N \frac{{{{( - 1)}^n}}}{{n!}} \sum\limits_{}^ *  {{\alpha _{{n_t}}}\frac{{{{( - 1)}^{{k_t}}}{{\left( {1 - {m_{{n_t}}}} \right)}_{{k_t}}}{{\left( {{\delta _{{n_t}}}} \right)}^{{k_t}}}{\lambda _{{\rm{RE}}}}}}{{{k_t}!{p_t}!{\vartheta_3^{A_3}}}}}  \notag\\
& \times \exp \left( { - \frac{{\sum\nolimits_{t = 1}^n \vartheta_3 {\gamma _{th}}\vartheta_2}}{\Psi } - {\lambda _{{\rm{RE}}}}\Delta (x)} \right)  {\left( {\frac{{{\gamma _{th}}}}{\Psi }} \right)^{A_2 }} \notag\\
& \times \int_0^{ + \infty } {{{\left( {\vartheta_2 + \frac{\vartheta_6}{z}} \right)}^{A_2 }}} \exp \left( { - \frac{{\sum\nolimits_{t = 1}^n {\vartheta_3} {\gamma _{th}}\vartheta_6}}{{\Psi z}} - {\lambda _{{\rm{RE}}}}z} \right)dz,
\end{align}
where $\vartheta_6 \triangleq {\Delta (x)\vartheta_2 + \Omega (x)}.$ By applying the Binomial Theorem, $Q$ can be calculated as
\begin{align}
	\label{EQ30}
Q &= 1 + \sum\limits_{n = 1}^N \frac{{{{( - 1)}^n}}}{{n!}}\notag\\& \sum\limits_{}^ *  {{\alpha _{{n_t}}}\frac{{{{( - 1)}^{{k_t}}}{{\left( {1 - {m_{{n_t}}}} \right)}_{{k_t}}}{{\left( {{\delta _{{n_t}}}} \right)}^{{k_t}}}{\lambda _{{\rm{RE}}}}}}{{{k_t}!{p_t}!}}{{\vartheta_3}^{ - A_3}}} \notag \\ 
& \times \exp \left( { - \frac{{\sum\nolimits_{t = 1}^n \vartheta_3 {\gamma _{th}}\vartheta_2}}{\Psi } - {\lambda _{{\rm{RE}}}}\Delta (x)} \right)  {\left( {\frac{{{\gamma _{th}}}}{\Psi }} \right)^{A_2 }} \notag  \\
&\times \sum\limits_{q = 0}^{A_2 } {\left( {\begin{array}{*{20}{c}}
			{A_2}\\
			q
	\end{array}} \right){\vartheta_2^{A_2 - q}}}  {\vartheta_6^q}\int_0^{ + \infty } {{z^{ - q}}} \times
 \notag\\
  &\exp \left( { - \frac{{\sum\nolimits_{t = 1}^n {\vartheta_3} {\gamma _{th}}\vartheta_6}}{{\Psi z}} - {\lambda _{{\rm{RE}}}}z} \right)dz.
\end{align}
\begin{lemma}
	\label{lemma_2}
	With using the same approach as \eqref{EQ24}, the closed-formed expression of $Q$ can be expressed as
	\begin{align}
		\label{EQ31}
		&Q = 1 + 2\sum\limits_{n = 1}^N \frac{{{{( - 1)}^n}}}{{n!}}  \sum\limits_{}^ *  {{\alpha _{{n_t}}}\frac{{{{( - 1)}^{{k_t}}}{{\left( {1 - {m_{{n_t}}}} \right)}_{{k_t}}}{{\left( {{\delta _{{n_t}}}} \right)}^{{k_t}}}{\lambda _{{\rm{RE}}}}}}{{{k_t}!{p_t}!{\vartheta_3^{A_3}}}}}  \notag\\
		& \times \exp \left( { - \frac{{\sum\nolimits_{t = 1}^n \vartheta_3 {\gamma _{th}}\vartheta_2}}{\Psi } - {\lambda _{{\rm{RE}}}}\Delta (x)} \right) 
		\notag\\& \times {\left( {\frac{{{\gamma _{th}}}}{\Psi }} \right)^{A_2 }}  \sum\limits_{q = 0}^{A_2 } {\left( {\begin{array}{*{20}{c}}
					A_2\\
					q
			\end{array}} \right){\vartheta_2^{A_2 - q}}}  {\vartheta_6^{\frac{{q + 1}}{2}}}{\left( {\frac{{\sum\nolimits_{t = 1}^n \vartheta_3 {\gamma _{th}}}}{{\Psi {\lambda _{{\rm{RE}}}}}}} \right)^{\frac{{ - q + 1}}{2}}} \notag \\
		& \times  {K_{ - q + 1}}\left( {2\sqrt {\frac{{\sum\nolimits_{t = 1}^n \vartheta_3 {\gamma _{th}}{\lambda _{{\rm{RE}}}}\vartheta_6}}{\Psi }} } \right).
	\end{align}
\end{lemma}

\begin{theorem}
	\label{theorem_2}
	Finally, by applying \cite[Eq.6.592.4]  {jeffrey2007table}, the closed-formed expression in terms of IP can be mathematically represented as
		\begin{align}
			\label{EQ40}
			&{\rm{IP}} = \sum\limits_{n = 1}^N \frac{{{{( - 1)}^{n + 1}}}}{{n!}} \times {\left( {\frac{{\sum\nolimits_{t = 1}^n \vartheta_3 {\gamma _{th}}}}{{\Psi {\lambda _{{\rm{RE}}}}}}} \right)^{\frac{{ - q + 1}}{2}}} \notag\\ &\times \sum\limits_{}^ *  {{\alpha _{{n_t}}}\frac{{{{( - 1)}^{{k_t}}}{{\left( {1 - {m_{{n_t}}}} \right)}_{{k_t}}}{{\left( {{\delta _{{n_t}}}} \right)}^{{k_t}}}{\lambda _{{\rm{RE}}}}{\lambda _{{\rm{JE}}}}}}{{{k_t}!{p_t}!{\vartheta_3^{A_3}}}}}  \notag\\
			& \times \exp \left( \Lambda_{10} \tilde \Lambda  - {\lambda _{{\rm{RE}}}}{{\mathord{\buildrel{\lower3pt\hbox{$\scriptscriptstyle\frown$}} 
							\over \gamma } }_{th}} - \frac{{\sum\nolimits_{t = 1}^n \vartheta_3 {\gamma _{th}}\vartheta_2}}{\Psi } \right)% \notag \\
			\end{align}
		    \begin{align*}
			& \times {\left( {\frac{{{\gamma _{th}}}}{\Psi }} \right)^{\sum\limits_{t = 1}^n {{p_t}} }} \sum\limits_{q = 0}^{\sum\limits_{t = 1}^n {{p_t}} } {\left( {\begin{array}{*{20}{c}}
						{\sum\limits_{t = 1}^n {{p_t}} }\\
						q
				\end{array}} \right)}  \frac{{{{\vartheta_2}^{A_2 - q}}{{\left( {\tilde \Lambda } \right)}^{\frac{{q + 3}}{2}}}}}{{\vartheta_7\vartheta_2{\gamma _{th}}}} \notag\\
			& \times \sum\limits_{w = 0}^\infty  \frac{{{{( - 1)}^w}{2^{q + 2w + 1}}{{\left( {\frac{{\left[ {{\lambda _{{\rm{JE}}}} + \frac{{{\gamma _{th}}{\lambda _{{\rm{RE}}}}\Theta }}{\Phi }} \right]\tilde \Lambda }}{{\vartheta_7\vartheta_2{\gamma _{th}}}}} \right)}^w}}}{{w!{\zeta ^{q + 2w + 1}}}} \notag\\ &\times G_{1,3}^{3,0}\left( {\frac{{{\zeta ^2}}}{4}\left| \begin{array}{l}
						0\\
						- 1,1 + w{\rm{,q + }}w
					\end{array} \right.} \right) ,
		\end{align*}
		where $\zeta  = 2\sqrt {\frac{{\sum\limits_{t = 1}^n \vartheta_3 {\gamma _{th}}{\lambda _{{\rm{RE}}}}\tilde \Lambda }}{\Psi }},$ and ${\rm{G}}_{p,q}^{m,n}\left( {z\left| \begin{array}{l}
					{a_1},\dots,{a_p}\\
					{b_1},\dots,{b_q}
				\end{array} \right.} \right)$
		is the Meijer G-function.
\end{theorem}
\begin{proof}
See Appendix~\ref{Appendix:theorem_2}.
\end{proof}
\begin{remark}
This paper derives the closed-form expressions for the IP and OP, independent of the small-scale fading coefficients and therefore working for a long time. We provide an initial mechanism to analyze the security and reliability of an integrated satellite-terrestrial multi-relay network under imperfect CSI. Since the analytical results in Theorems ~\ref{theorem_1}  and \ref{theorem_2} are only the multivariate functions of channel statistics, the IP and OP can be evaluated at a lower cost than numerically averaging over many different realizations of the small-scale fading coefficients. This fact reduces the cost of the system design and can be easily  deployed in practice.
\end{remark}
\begin{remark}
The system-level design makes steps toward practical applications. One of the possible scenarios is that a satellite provides mobile streaming services and the Internet to ordinary handheld devices located in indoor areas such as multi-floor buildings. Under this scenario, the end device can certainly not receive signals from the satellite directly. As a result, a terrestrial network and a set of gap-filler are necessary to help the end device fully access these services \cite{etsi3102,nguyen2022outage}.
\end{remark}
\begin{figure*}[t]
	\centering
	\begin{minipage}{.48\textwidth}
		\centering
		\includegraphics[width=8.5cm,height=7cm]{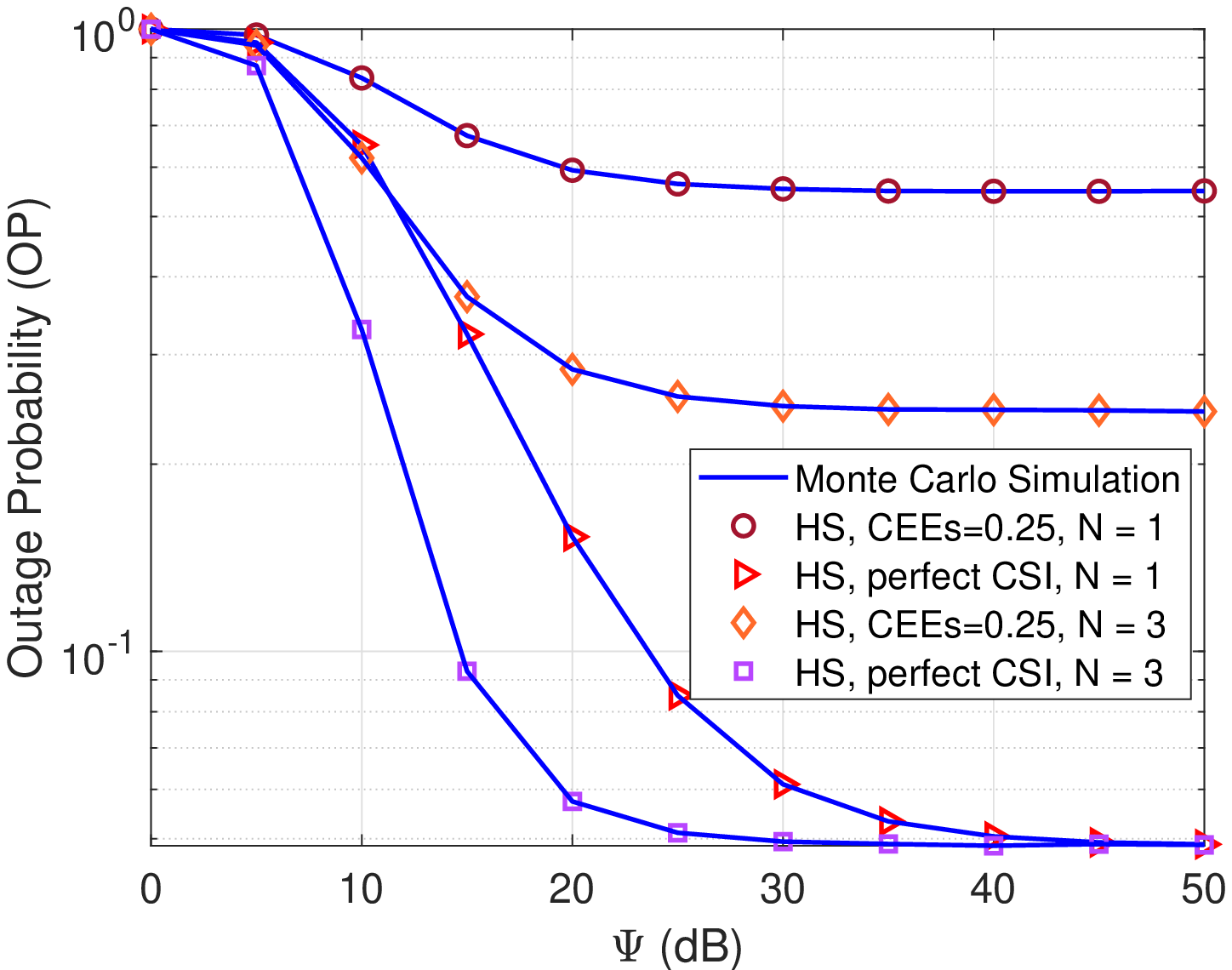}
		\caption{OP in HS case versus $\Psi$ with $\rm C_{th}$ = 1(bps/Hz) and $\Phi$=20 dB.}
		\label{fig:2}
			\vspace{-0.2cm}
	\end{minipage} \hfill
	\begin{minipage}{.48\textwidth}
		\centering
		\includegraphics[width=8.5cm,height=7cm]{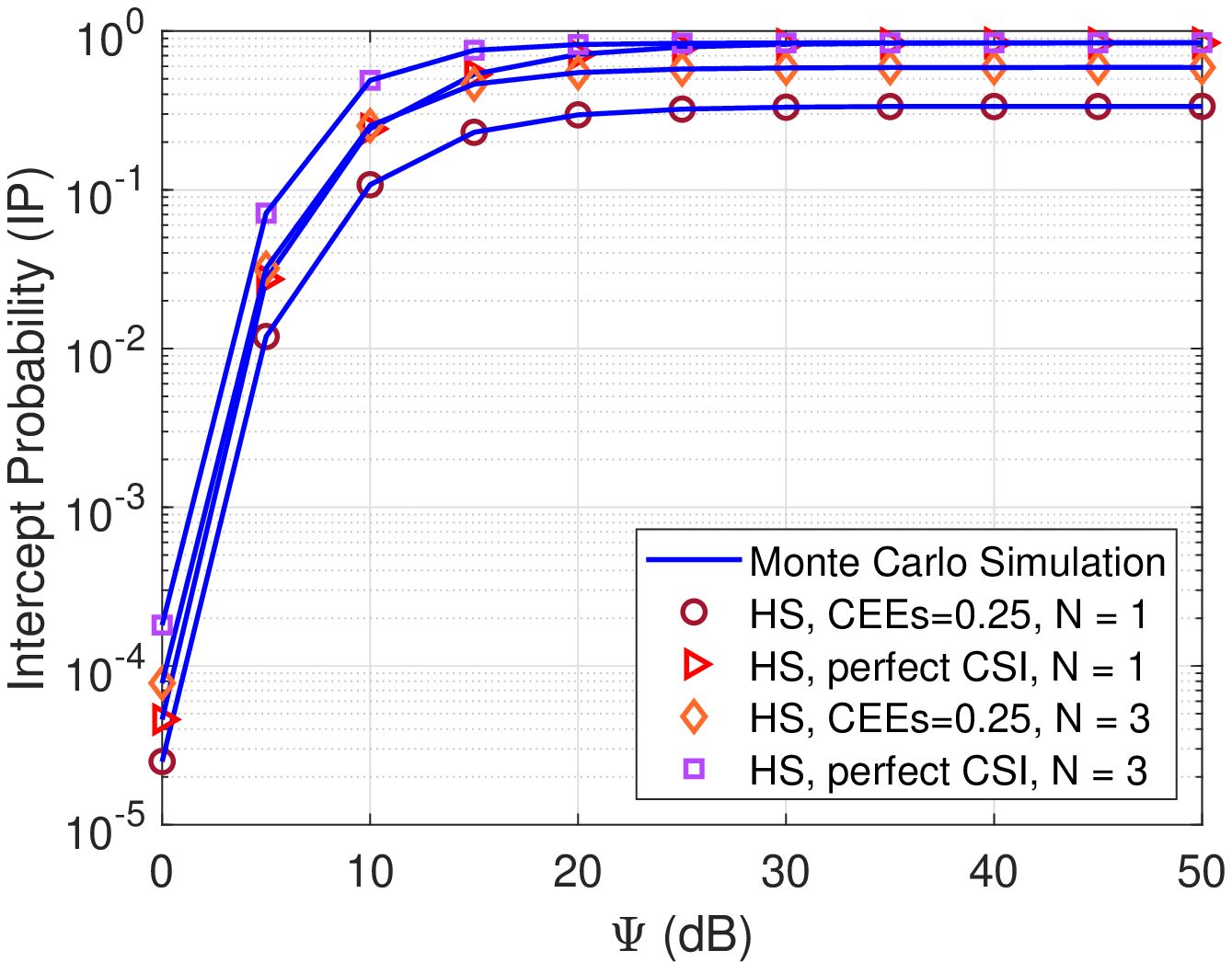}
		\caption{IP in HS versus $\Psi$ with $\rm C_{th}$ = 1(bps/Hz), $\Phi$=20 dB and $\Theta$=1 dB.}
		\label{fig:3}
		\vspace{-0.2cm}
	\end{minipage}
\vspace*{-0.2cm}
\end{figure*}

\begin{figure*}[t]
	\centering
	\begin{minipage}{.48\textwidth}
		\centering
		\includegraphics[width=8.5cm,height=7cm]{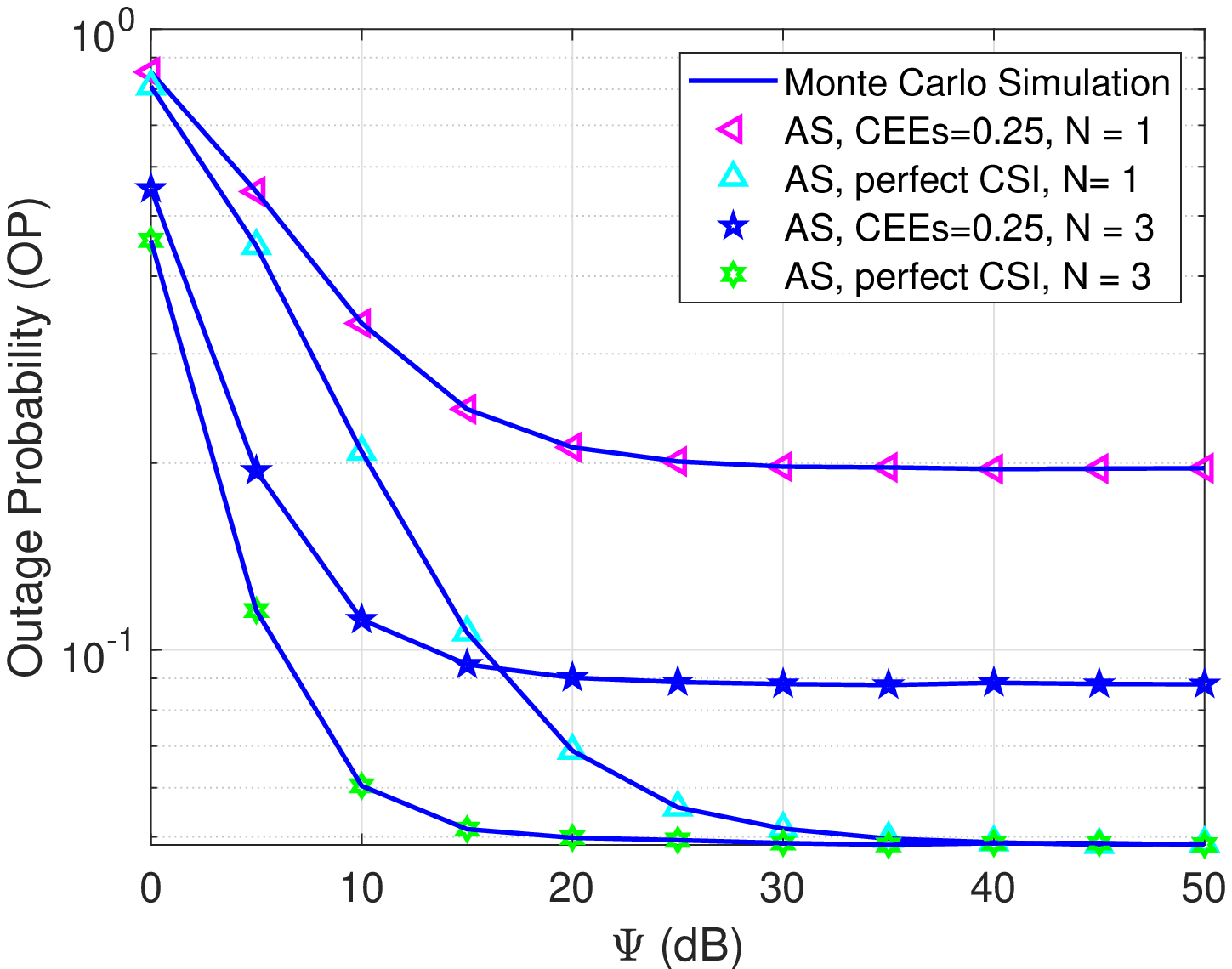}
		\caption{OP in AS case versus $\Psi$ with $\rm C_{th}$ = 1(bps/Hz) and $\Phi$=20 dB.}
		\label{fig:4}
	\end{minipage} \hfill
	\begin{minipage}{.48\textwidth}
		\centering
		\includegraphics[width=8.5cm,height=7cm]{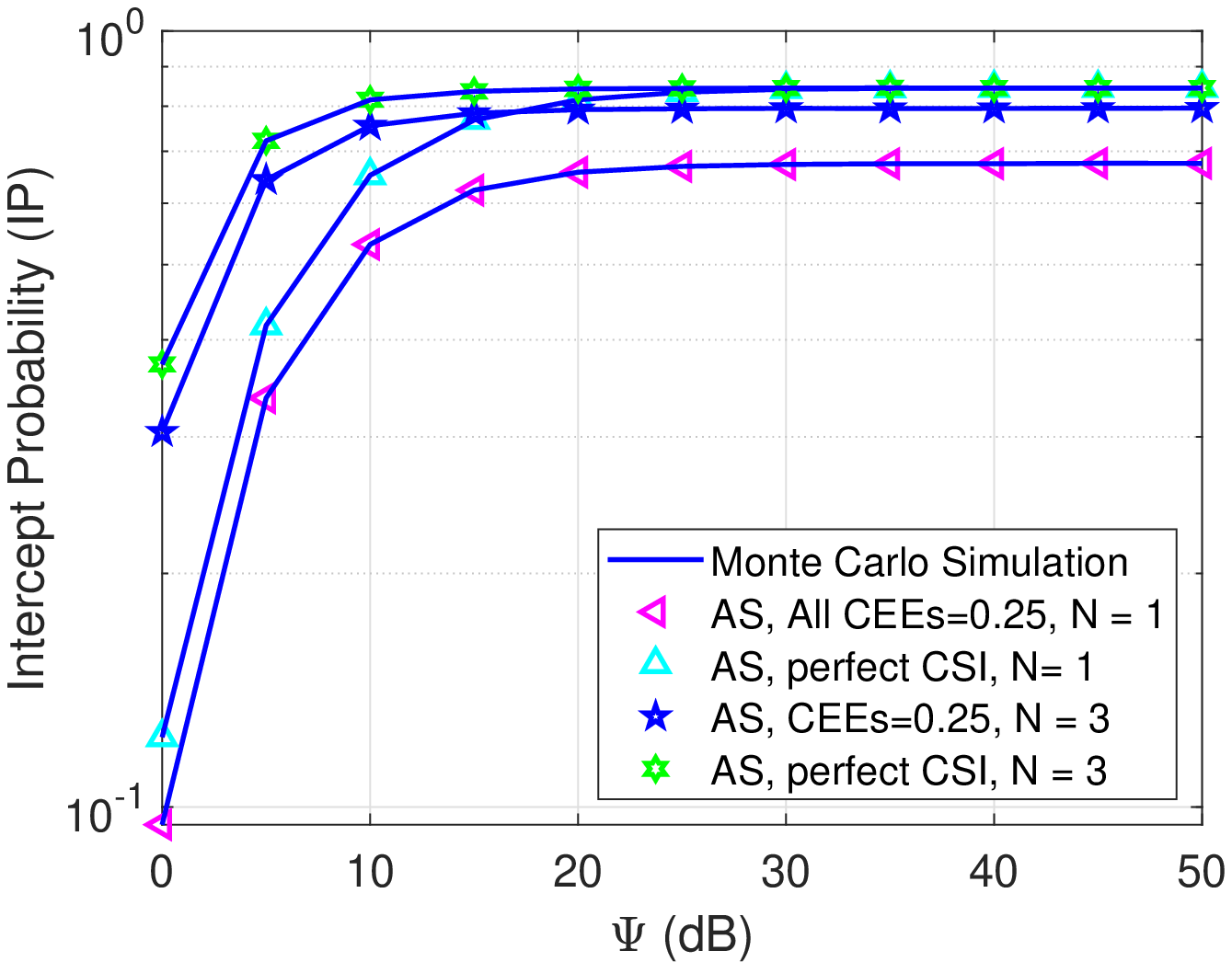}
		\caption{IP in AS versus $\Psi$ with $\rm C_{th}$ = 1(bps/Hz), $\Phi$=20 dB and $\Theta$=1 dB. }
		\label{fig:5}
	\end{minipage}
\vspace*{-0.4cm}
\end{figure*}

\begin{figure*}[t]
	\centering
	\begin{minipage}{.48\textwidth}
		\centering
		\includegraphics[width=8.5cm,height=7cm]{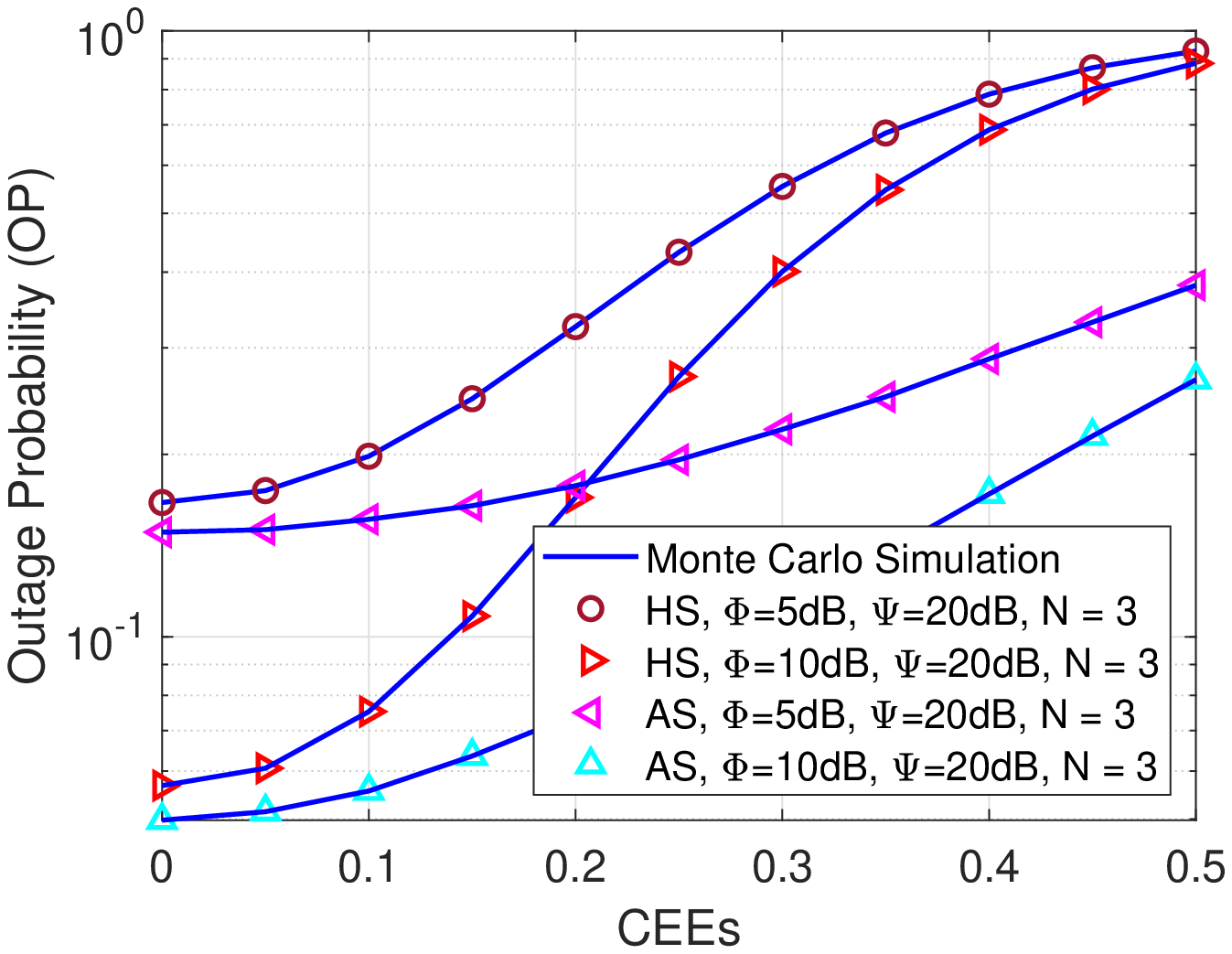}
		\caption{OP versus all CEEs with N= 3, $\rm C_{th}$=1 bps/Hz and $\Psi$ = 20 dB.}
		\label{fig:6}
	\end{minipage} \hfill
	\begin{minipage}{.48\textwidth}
		\centering
		\includegraphics[width=8.5cm,height=7cm]{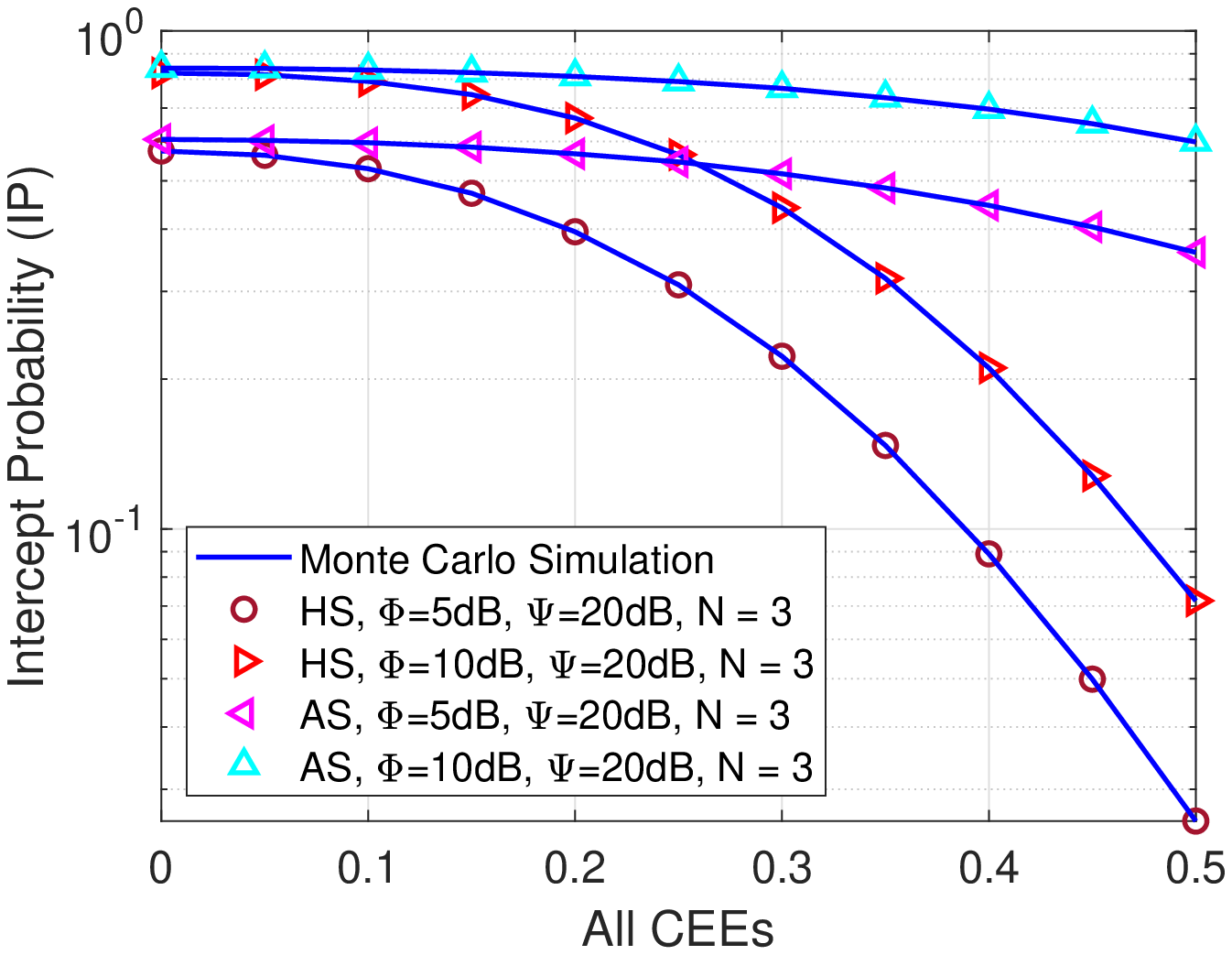}
		\caption{IP versus all CEEs with N= 3, $\rm C_{th}$=1 bps/Hz, $\Psi$ = 20 dB, $\Phi$=5 dB and $\Theta$= 1 dB.}
		\label{fig:7}
	\end{minipage}
\vspace*{-0.4cm}
\end{figure*}

\begin{figure*}[t]
	\centering
	\begin{minipage}{.48\textwidth}
		\centering
		\includegraphics[width=8.5cm,height=7cm]{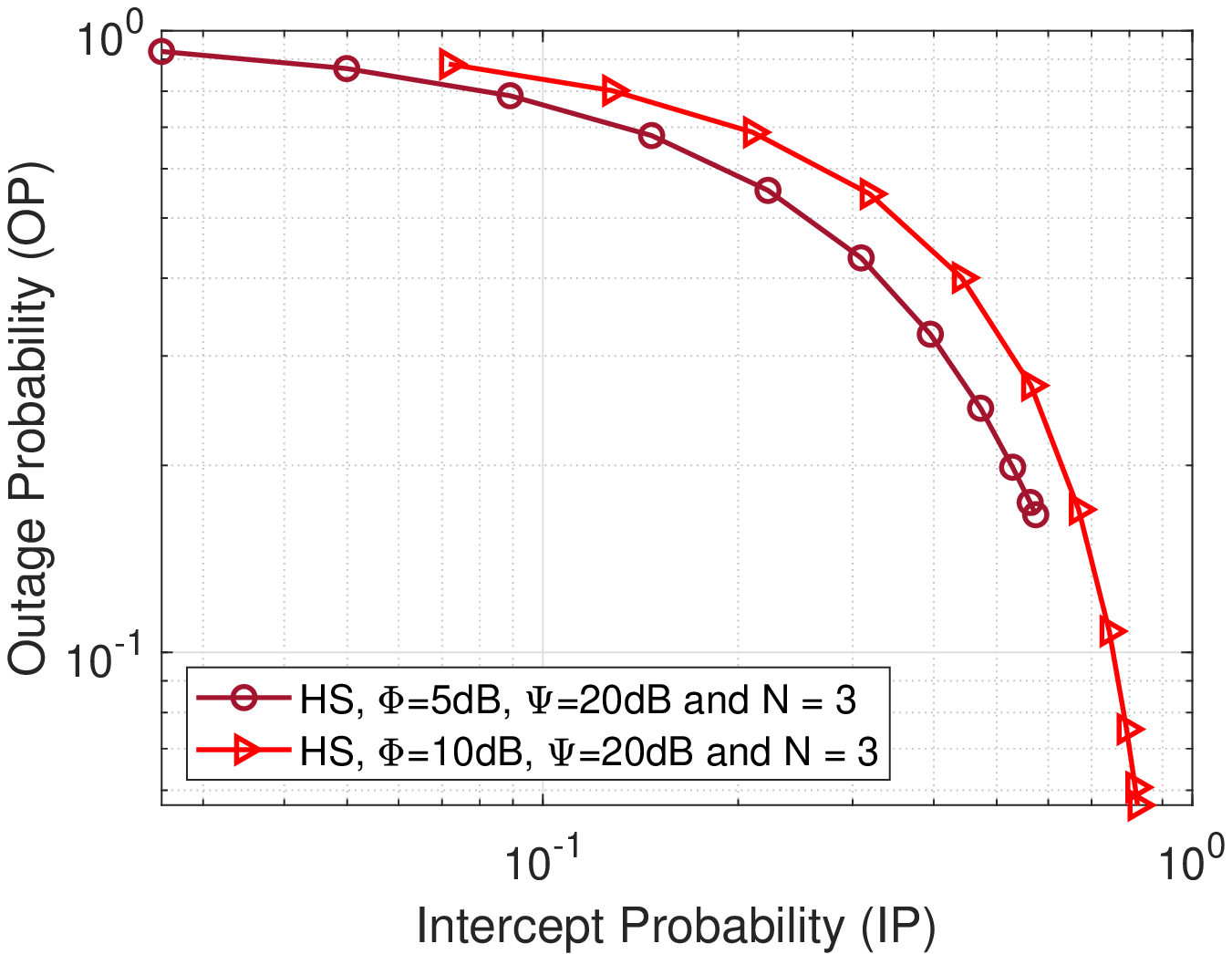}
		\caption{OP versus IP in HS case with $\Psi$ = 20 dB, $\rm C_{th}$ = 1 (bps/Hz), N=3 and $\Theta$=1 dB.}
		\label{fig:8}
	\end{minipage} \hfill
	\begin{minipage}{.48\textwidth}
		\centering
		\includegraphics[width=8.5cm,height=7cm]{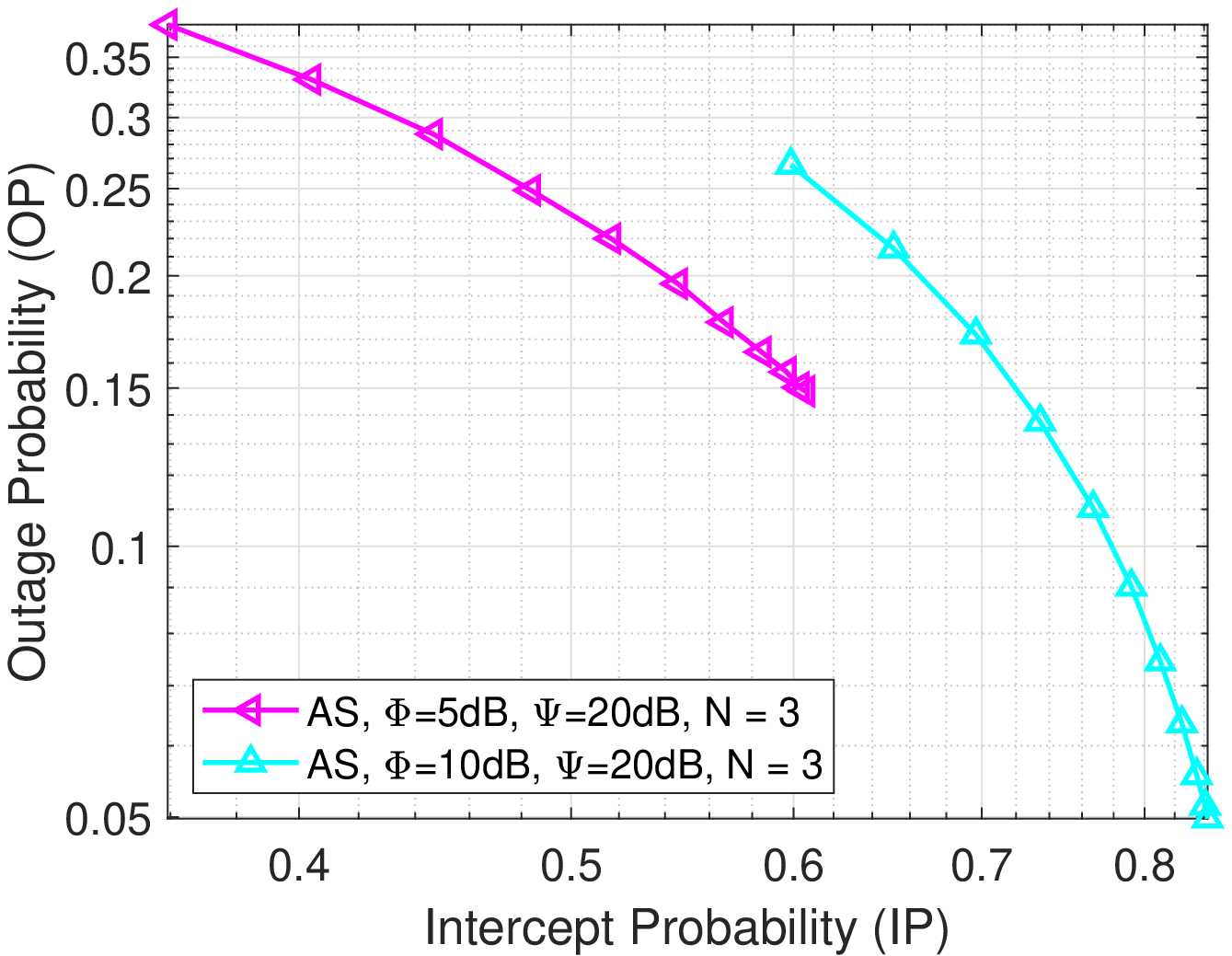}
		\caption{OP versus IP in AS case with $\Psi$ = 20 dB, $\rm C_{th}$ = 1 (bps/Hz), N=3 and $\Theta$=1 dB.}
		\label{fig:9}
	\end{minipage}
\vspace*{-0.4cm}
\end{figure*}

\vspace{-0.2cm}
\section{Numerical Results}
\label{sec:4}

In this section, Monte-Carlo simulations are presented to evaluate the efficiency of the proposed HSTRN and the influences of various parameters on the system performance. The simulation parameters are given in Table \ref{Table_2} \cite{3gpp2019study} and \cite{yue2020outage}. {Unless other stated, we assume the SNR threshold of the system is set as $C_{\rm th}=1$ bps/Hz, (m,b,$\Omega$)=(1,0.0635,0.0007) for heavy shadowing (HS) case and (m,b,$\Omega$)=(5,0.25,0.279) for average shadowing (AS) case}. For
ease of notation and clear observation, we denote methods with CEEs = 0.25-N = 1, perfect CSI-N = 1, CEEs = 0.25-N = 3, perfect CSI-N = 3, imperfect CSI-$\Phi$ = 5 dB-$\Psi$ = 20 dB-N = 3, and imperfect CSI-$\Phi$ = 10 dB-$\Psi$ = 20 dB-N = 3 are proposed method 1 (PM1), proposed method 2 (PM2), proposed method 3 (PM3), proposed method 4 (PM4), proposed method 5 (PM5), and proposed method 6 (PM6), respectively.

In Figs. \ref{fig:2} and \ref{fig:3}, we investigate the influences of satellite's transmit power to the outage and intercept probability with heavy shadowing, respectively. {The setting parameters for Figs. \ref{fig:2} and \ref{fig:3} are given as $\rm C_{th} =$ 1 bps/Hz and $\Phi$=20 dB.} As given in Fig. \ref{fig:2}, the OP is significantly enhanced with a higher value of $\Psi$, i.e., $\Psi$ from 0 to 50 dB. We stress that the value $\Psi$ reaches $50$~dB, which is impractical, but the range is selected to demonstrate the convergence and to have a full vision of the OP and IP.  For example, the OP values of PM1 with high shadowing are 0.9783, 0.6742, and 0.5639 corresponding to $\Psi$ equals 10, 20, and 30 dB, respectively. Moreover, we can observe from Fig. \ref{fig:2} that the OP value is decreased to a saturated value with the increasing of $\Psi$. For example, when $\Psi$ values ranges from $35 \to 50$ dB, the OP values of PM1 and PM3 converge to 0.5495 and 0.2430, respectively. This is because a higher $\Psi$ value, the more the satellite's transmit power is. Therefore, the receiver D has more chance to successfully decode the received signals from relay ${\rm R}_n$. Furthermore, the outage performance is significantly improved with a higher number of relay. This is because the satellite can select a better relay with higher channel gain from ${\rm R}_n \to D$, which can improve the total throughput received at the destination.

In Fig. \ref{fig:3}, we can see that the IP is enhanced with a higher value of $\Psi$. More specifically, the IP values of PM1 with high shadowing (HS) are 0.0119, 0.1077, and 0.2305 corresponding to $\Psi$ values equals 5, 10, and 15 dB, respectively. That is because a better signal-to-noise ratio (SNR) can be obtained at the eavesdropper with higher values of $\Psi$, as shown in Eq. \eqref{EQ12}. It also can be observed that the intercept performance of the proposed schemes with perfect CSI is better than that compared to imperfect CSI with a higher value of $\Psi$. For instance, at $\Psi$ equals 50 dB, the IP values of the PM1 and PM4 are 0.3357 and 0.8442, respectively. This is because the eavesdropper can perfectly decode received signals with perfect CSI, which greatly improves the intercept performance. 

Figs. \ref{fig:4} and \ref{fig:5} illustrate the influences of the satellite's transmit power to the OP and IP with average shadowing (AS) effects, respectively. As observed from Fig. \ref{fig:4}, the outage values of the proposed schemes with perfect CSI, i.e., PM2 and PM4, are better than that as compared to imperfect CSI ones, i.e, PM1 and PM3. Specifically, when $\Psi$ equals 40 dB, the outage values of PM1 and PM3 impose 0.1956 and 0.049, respectively. Moreover, the proposed schemes converge to saturation value faster with higher number of relays, i.e., N from 1 to 3. For example, the PM2 and PM4 schemes converge to saturation value at $\Psi$ equals 40 and 30 dB, respectively. In contrast to Fig. \ref{fig:4}, the intercept performance in Fig. \ref{fig:5} is greatly increased with higher value of satellite's transmit power. Specifically, when $\Psi$ increases from 10 to 20 dB, the IP of the PM1 is 0.5310 and 0.6579, respectively. Similar to Fig. \ref{fig:3}, the intercept performance of proposed schemes with perfect CSI outperforms imperfect ones.\footnote{Concerning the values of the OP, some reach one under certain conditions, and it is, of course, ideal for many practical applications under a finite resource budget. The rationale behind this is to justify the help of relays to enhance security and reliability in satellite-terrestrial networks once sufficient radio resources are available. We stress that there are parameter settings that increase the number of relays the OP dramatically decreases. The OP facilitates approximately 10-fold as the number of relays increases from 1 to 3, as shown in Figs. \ref{fig:2} and \ref{fig:4}.}

In Figs. \ref{fig:6} and \ref{fig:7}, we show the influences of CEEs to the outage and intercept performance with heavy and average showing effects, respectively. {The setting parameters for Figs. \ref{fig:6} and \ref{fig:7} are given as $\rm C_{th} =$ 1 bps/Hz, N=3, $\Psi$=20 dB, $\Phi$=5 dB and $\Theta$=1 dB.} First, we can observe from Fig. \ref{fig:6} that the outage performance of the HS-based schemes is worse than that compared to the AS-based schemes. Specifically, at CEEs equals 0.1, PM5 with HS and PM5 with AS impose 0.1983 and 0.1563, respectively. Moreover, the gap between PM5 with HS and PM5 with AS is larger with higher of CEEs values. For instance, the gap between these schemes are 0.0423, 0.1475, and 0.2354 at CEEs equals 0.1, 0.2, and 0.3, respectively. This shows significant effects of the shadowing effects on network performance. Second, it is also can be seen that the outage performance is enhanced by increasing $\phi$ value. This can be explained by the fact that the higher the $\phi$ value is, the better SNR at the destination can be obtained, as shown in Eq. \eqref{EQ11}. It can be seen from Fig. \ref{fig:7} that the intercept performance is significantly degraded with higher CEEs  values. For instance, when CEEs is between 0.1 and 0.3, the IP value of PM5 with HS is 0.529, 0.3951, and 0.2222, respectively. Moreover, the intercept performance of the HS-based schemes is significantly degraded as compared to AS-based ones, which is in contrast to Fig. \ref{fig:6}.
 
In Figs. \ref{fig:8} and \ref{fig:9}, we investigate the security-reliability trade-off {with parameters are setting as Figs.      \ref{fig:3} and \ref{fig:5}}. It is clear to see that for any specific IP, the OP of PM5 with HS (or AS) is higher than that compared to PM6 with HS (or AS). For example, when IP equals 0.4, the OP of the PM5 with HS is 0.315, while the OP of the PM6 with HS is 0.422. Additionally, Figs. \ref{fig:8} and \ref{fig:9} also show that PM6 scheme with HS obtains the best OP compared to other HS-based and AS-based ones.

\section{Conclusion and Discussion}
\label{sec:5}
This paper has studied the physical layer security in terms of security and reliability trade-off of a satellite amplify-and-forward multi-relay network, consisting of one satellite, multiple relays, one destination in the presence of one eavesdropper. Especially, a jammer has been implemented to create artificial noise in order to eliminate the influence of the eavesdropper on system performance. Based on this system model, the closed-formed expressions of OP and IP were derived by considering generalized Shadowed Rician and Rayleigh fading in the first and second hop, respectively. Then, these mathematical analyses have been validated through simulation results. In particular, the values of the satellite's transmit power, number of relays, and relay transmit power can be selected appropriately to reduce the influences of the eavesdropper.

\vspace*{-0.2cm}
\appendix
\subsection{Proof of Theorem~\ref{theorem_1}}
\label{Appendix:theorem_1}
Based on \eqref{EQ11} and \eqref{EQ19}, the OP can be rewritten as follows:
\begin{align}
	\label{EQ20}
	{\rm{OP}} &= \rm{Pr} \left( {\frac{{\Psi \Phi {\gamma _{{\rm{S}}{{\rm{R}}_a}}}{\gamma _{{{\rm{R}}_a}{\rm{D}}}}}}{{{\gamma _{{{\rm{R}}_a}{\rm{D}}}}\Phi \vartheta_2 + {\gamma _{{\rm{S}}{{\rm{R}}_a}}}\Psi \vartheta_4 + \Xi }} < {\gamma _{th}}} \right)\notag\\
	& = \Pr \left( {\gamma _{{\rm{S}}{{\rm{R}}_a}}}\Lambda_3 < {\gamma _{th}}\left[ {{\gamma _{{{\rm{R}}_a}{\rm{D}}}}\Phi \vartheta_2 + \Xi } \right] \right)%\notag\\
\end{align}
\begin{align*}
  & = \left\{ \begin{array}{l}
		\Pr \left( {\gamma _{{\rm{S}}{{\rm{R}}_a}}} < \frac{\Lambda_2}{\Lambda_3} \right),{\rm{if}}\,{\gamma _{{{\rm{R}}_a}{\rm{D}}}} > \frac{{{\gamma _{th}} \vartheta_4}}{\Phi }\\
		1\hfill,{\rm{if}}\,{\gamma _{{{\rm{R}}_a}{\rm{D}}}} \le \frac{{{\gamma _{th}} \vartheta_4 }}{\Phi }
	\end{array} \right. \notag\\
	& = \int_0^{{{\tilde \gamma }_{th}}} {{f_{{\gamma _{{{\rm{R}}_a}{\rm{D}}}}}}(x)dx}  + \int_{{{\tilde \gamma }_{th}}}^{ + \infty } {{F_{{\gamma _{{\rm{S}}{{\rm{R}}_a}}}}}} \left( \frac{\Lambda_4}{\Lambda_5}  \right)  {f_{{\gamma _{{{\rm{R}}_a}{\rm{D}}}}}}(x)dx,
\end{align*}
where ${\tilde \gamma _{th}} \triangleq  \frac{{{\gamma _{th}}\left( {\mu _{{\rm{RD}}}^2\Phi  + 1} \right)}}{\Phi },$ $\vartheta_4 \triangleq  \left( {\mu _{{\rm{RD}}}^2\Phi  + 1} \right),$ $\Lambda_2 \triangleq \left[ {{\gamma _{{{\rm{R}}_a}{\rm{D}}}}\Phi \vartheta_2 + \Xi } \right]$, $\Lambda_3 \triangleq \Psi \left[ {\Phi \lambda_{\rm R_a D} - {\gamma _{th}}\vartheta_4} \right]$, $\Lambda_4 \triangleq \left[ {x \Phi \vartheta_2 + \Xi } \right]$, $\Lambda_5 \triangleq \Psi \left[ {\Phi x- {\gamma _{th}}\vartheta_4} \right]$. By substituting \eqref{EQ18} into \eqref{EQ20}, the OP can be mathematically represented by
\begin{align}
	\label{EQ21}
	{\rm{OP}} &= 1 + \sum\limits_{n = 1}^N \frac{{{{( - 1)}^n}}}{{n!}} \sum\limits_{}^ *  {{\alpha _{{n_t}}}\frac{{{{( - 1)}^{{k_t}}}{{\left( {1 - {m_{{n_t}}}} \right)}_k}{{\left( {{\delta _{{n_t}}}} \right)}^{{k_t}}}{\lambda _{{\rm{RD}}}}}}{{{k_t}!{p_t}!{{\vartheta_3}^{A_3}}}}}  \notag\\
	& \times \int_{{{\tilde \gamma }_{th}}}^{ + \infty } {{\left( {\frac{{{\gamma _{th}}\left[ {x \vartheta_2 + \frac{\Xi }{\Phi }} \right]}}{{\Psi \left[ {x - {{\tilde \gamma }_{th}}} \right]}}} \right)}^{\sum\nolimits_{t = 1}^n {{p_t}} }} \notag\\ & \times \exp \left( { - \sum\limits_{t = 1}^n {\vartheta_3\left( {\frac{{{\gamma _{th}}\left[ {x \vartheta_2 + \frac{\Xi }{\Phi }} \right]}}{{\Psi \left[ {x - {{\tilde \gamma }_{th}}} \right]}}} \right) - {\lambda _{{\rm{RD}}}}x} } \right) dx.
\end{align}
By denoting $y = x - {\tilde \gamma _{th}}$, \eqref{EQ21} can be further reformulated as
\begin{align}
	\label{EQ22}
	&{\rm{OP}} = 1 + \sum\limits_{n = 1}^N \frac{{{{( - 1)}^n}}}{{n!}}  \sum\limits_{}^ *  {{\alpha _{{n_t}}}\frac{{{{( - 1)}^{{k_t}}}{{\left( {1 - {m_{{n_t}}}} \right)}_{{k_t}}}{{\left( {{\delta _{{n_t}}}} \right)}^{{k_t}}}{\lambda _{{\rm{RD}}}}}}{{{k_t}!{p_t}!{\vartheta_3^{A_3}}}}}  \notag\\
	& \times \int_0^{ + \infty } {{\left( {\frac{{{\gamma _{th}}\left[ {\left( {y + {{\tilde \gamma }_{th}}} \right)\left( {\mu _{{\rm{SR}}}^2\Psi  + 1} \right) + \frac{\Xi }{\Phi }} \right]}}{{\Psi y}}} \right)}^{\sum\nolimits_{t = 1}^n {{p_t}} }} \notag \\
	&\times \exp \left( \begin{array}{l}
		- \sum\limits_{t = 1}^n {\vartheta_3 \left( {\frac{{{\gamma _{th}}\left[ {\left( {y + {{\tilde \gamma }_{th}}} \right) \vartheta_2 + \frac{\Xi }{\Phi }} \right]}}{{\Psi y}}} \right)} \\
		- {\lambda _{{\rm{RD}}}}\left( {y + {{\tilde \gamma }_{th}}} \right)
	\end{array} \right) dy  \notag\\
	& = {\rm{1 + }} \sum\limits_{n = 1}^N \frac{{{{( - 1)}^n}}}{{n!}}\sum\limits_{}^ *  {\alpha _{{n_t}}} \frac{{{{( - 1)}^{{k_t}}}{{\left( {1 - {m_{{n_t}}}} \right)}_{{k_t}}}{{\left( {{\delta _{{n_t}}}} \right)}^{{k_t}}}{\lambda _{{\rm{RD}}}}}}{{{k_t}!{p_t}!{\vartheta_3^{A_3}}}}  \notag \\
	& \times {\left( {\frac{{{\gamma _{th}}}}{\Psi }} \right)^{A_2 }}\exp \left( { - \sum\limits_{t = 1}^n {\vartheta_3 \frac{{{\gamma _{th}}\vartheta_2}}{\Psi }}  - {\lambda _{{\rm{RD}}}}{{\tilde \gamma }_{th}}} \right) \notag\\
	& \times \int_0^{ + \infty } {{{\left( {\vartheta_2 + \frac{{\tilde \Xi }}{y}} \right)}^{\sum\limits_{t = 1}^n {{p_t}} }}}  \exp \left( { - \frac{{\sum\nolimits_{t = 1}^n {\vartheta_3\frac{{{\gamma _{th}}\tilde \Xi }}{\Psi }} }}{y} - {\lambda _{{\rm{RD}}}}y} \right)dy,
\end{align}
where $	\tilde \Xi  = \left[ {{{\tilde \gamma }_{th}}\left( {\mu _{{\rm{SR}}}^2\Psi  + 1} \right) + \frac{\Xi }{\Phi }} \right] $. By applying Binomial Theorem, i.e., ${(x + y)^m} = \sum\limits_{q = 0}^m {\left( \begin{array}{l}
		m\\
		q
	\end{array} \right)} {x^{m - q}}{y^q}$, \eqref{EQ22} can be  rewritten as
\begin{align}
	\label{EQ23}
	{\rm{OP}} &= 1 + \sum\limits_{n = 1}^N \frac{{{{( - 1)}^n}}}{{n!}}  \sum\limits_{}^ *  {{\alpha _{{n_t}}}\frac{{{{( - 1)}^{{k_t}}}{{\left( {1 - {m_{{n_t}}}} \right)}_{{k_t}}}{{\left( {{\delta _{{n_t}}}} \right)}^{{k_t}}}{\lambda _{{\rm{RD}}}}}}{{{k_t}!{p_t}!{\vartheta_3^{A_3}}}}}  \notag \\
	& \times {\left( {\frac{{{\gamma _{th}}}}{\Psi }} \right)^{A_2 }}\exp \left( { - \sum\limits_{t = 1}^n {\vartheta_3\frac{{{\gamma _{th}}\vartheta_2}}{\Psi }}  - {\lambda _{{\rm{RD}}}}{{\tilde \gamma }_{th}}} \right) \notag
\end{align}
\begin{align}
	& \times \sum\limits_{q = 0}^{A_2} {\left( \begin{array}{l}
			A_2\\
			q
		\end{array} \right)} {\vartheta_2^{A_2 - q}}{{\tilde \Xi }^q}\int_0^{ + \infty } {{y^{ - q}}} \notag\\
	 &\exp \left( { - \frac{{\sum\nolimits_{t = 1}^n {\vartheta_3 \frac{{{\gamma _{th}}\tilde \Xi }}{\Psi }} }}{y} - {\lambda _{{\rm{RD}}}}y} \right)dy.
\end{align}
By adopting \cite[Eq.3.471.9]{jeffrey2007table}, the analytical expression of the OP can be obtained in the theorem after doing some algebra.
\vspace{-0.2cm} 
\subsection{Proof of Theorem~\ref{theorem_2}}
\label{Appendix:theorem_2}
Based on the result obtained from Lemma \ref{lemma_2}, by substituting \eqref{EQ31} into \eqref{EQ26}, the IP can be calculated as
\begin{align}
	\label{EQ32}
	{\rm{IP}} &= 2\sum\limits_{n = 1}^N \frac{{{{( - 1)}^{n + 1}}}}{{n!}} \sum\limits_{}^ *  {{\alpha _{{n_t}}}\frac{{{{( - 1)}^{{k_t}}}{{\left( {1 - {m_{{n_t}}}} \right)}_{{k_t}}}{{\left( {{\delta _{{n_t}}}} \right)}^{{k_t}}}{\lambda _{{\rm{RE}}}}{\lambda _{{\rm{JE}}}}}}{{{k_t}!{p_t}!{\vartheta_3^{A_3}}}}} %\notag \\
	\end{align}
\begin{align*}
	&\times {\left( {\frac{{\sum\nolimits_{t = 1}^n \vartheta_3 {\gamma _{th}}}}{{\Psi {\lambda _{{\rm{RE}}}}}}} \right)^{\frac{{ - q + 1}}{2}}} \exp \left( { - \frac{{\sum\nolimits_{t = 1}^n \vartheta_3 {\gamma _{th}}\vartheta_2}}{\Psi }} \right)
	\notag\\
	 &\times {\left( {\frac{{{\gamma _{th}}}}{\Psi }} \right)^{A_2 }}  \sum\limits_{q = 0}^{\sum\nolimits_{t = 1}^n {{p_t}} } {\left( {\begin{array}{*{20}{c}}
				A_2\\
				q
		\end{array}} \right){\vartheta_2^{A_2 - q}}} \notag \\
	&\times \int_0^{ + \infty } {{\vartheta_6^{\frac{{q + 1}}{2}}}\exp \left( { - {\lambda _{{\rm{JE}}}}x - {\lambda _{{\rm{RE}}}}\Delta (x)} \right)} \notag\\ & \times {K_{ - q + 1}}\left( {2\sqrt {\frac{{\sum\nolimits_{t = 1}^n \vartheta_3 {\gamma _{th}}{\lambda _{{\rm{RE}}}}\vartheta_6}}{\Psi }} } \right)dx.
\end{align*}
Based on \eqref{EQ32}, we have 
\begin{align}
	\label{EQ33}
	&\Delta (x) \vartheta_2 + \Omega (x)  = {{\mathord{\buildrel{\lower3pt\hbox{$\scriptscriptstyle\frown$}} 
				\over \gamma } }_{th}}\vartheta_2 + \frac{{{\gamma _{th}}\Theta x}}{\Phi }\vartheta_2 + \frac{{x\Theta \vartheta_2 + \Lambda }}{\Phi } \notag\\
	&= \vartheta_2 \left( {{{\mathord{\buildrel{\lower3pt\hbox{$\scriptscriptstyle\frown$}} 
					\over \gamma } }_{th}} + \frac{{{\gamma _{th}}\Theta x}}{\Phi } + \frac{{x\Theta }}{\Phi }} \right) + \frac{\Lambda }{\Phi },\\
	\label{EQ34}
	&\exp \left( { - {\lambda _{{\rm{JE}}}}x - {\lambda _{{\rm{RE}}}}\Delta (x)} \right) \notag\\ &= \exp \left( { - {\lambda _{{\rm{JE}}}}x - {\lambda _{{\rm{RE}}}}{{\mathord{\buildrel{\lower3pt\hbox{$\scriptscriptstyle\frown$}} 
					\over \gamma } }_{th}} - \frac{{{\gamma _{th}}{\lambda _{{\rm{RE}}}}\Theta x}}{\Phi }} \right) \notag\\
	& = \exp \left( { - {\lambda _{{\rm{RE}}}}{{\mathord{\buildrel{\lower3pt\hbox{$\scriptscriptstyle\frown$}} 
					\over \gamma } }_{th}}} \right)\exp \left( { - x\left[ {{\lambda _{{\rm{JE}}}} + \frac{{{\gamma _{th}}{\lambda _{{\rm{RE}}}}\Theta }}{\Phi }} \right]} \right).
\end{align}
Next, by denoting $t \triangleq \frac{{x\left[ {\left( {\frac{\Theta }{\Phi } + 1} \right)\vartheta_2{\gamma _{th}}} \right]}}{{\tilde \Lambda }} + 1$, \eqref{EQ33} and \eqref{EQ34} can be respectively rewritten as
\begin{align}
	\label{EQ35}
	&\Delta (x)\vartheta_2 + \Omega (x) = t\tilde \Lambda,\\
	\label{EQ36}
	&\exp \left( { - {\lambda _{{\rm{JE}}}}x - {\lambda _{{\rm{RE}}}}\Delta (x)} \right) \notag\\ &= \exp \left( {\frac{{\tilde \Lambda \left[ {{\lambda _{{\rm{JE}}}} + \frac{{{\gamma _{th}}{\lambda _{{\rm{RE}}}}\Theta }}{\Phi }} \right]}}{{\vartheta_7\vartheta_2{\gamma _{th}}}} - {\lambda _{{\rm{RE}}}}{{\mathord{\buildrel{\lower3pt\hbox{$\scriptscriptstyle\frown$}} 
					\over \gamma } }_{th}}} \right) \notag\\
	& \times \exp \left( { - \frac{{\left[ {{\lambda _{{\rm{JE}}}} + \frac{{{\gamma _{th}}{\lambda _{{\rm{RE}}}}\Theta }}{\Phi }} \right]t\tilde \Lambda }}{{\vartheta_7\vartheta_2{\gamma _{th}}}}} \right),
\end{align}
where $\tilde \Lambda  \triangleq {\tilde \gamma _{th}}\vartheta_2 + \frac{\Lambda }{\Phi },$
$\vartheta_7 \triangleq {\frac{\Theta }{\Phi } + 1}.$ By substituting \eqref{EQ35} and \eqref{EQ36} into \eqref{EQ32}, it yields
\begin{align}
	\label{EQ37}
	{\rm{IP}} &= 2\sum\limits_{n = 1}^N \frac{{{{( - 1)}^{n + 1}}}}{{n!}} {\left( {\frac{{\sum\nolimits_{t = 1}^n \vartheta_3 {\gamma _{th}}}}{{\Psi {\lambda _{{\rm{RE}}}}}}} \right)^{\frac{{ - q + 1}}{2}}}\notag\\ &\sum\limits_{}^ *  {{\alpha _{{n_t}}}\frac{{{{( - 1)}^{{k_t}}}{{\left( {1 - {m_{{n_t}}}} \right)}_{{k_t}}}{{\left( {{\delta _{{n_t}}}} \right)}^{{k_t}}}{\lambda _{{\rm{RE}}}}{\lambda _{{\rm{JE}}}}}}{{{k_t}!{p_t}!{\vartheta_3^{A_3}}}}}   
\end{align}
\begin{align*}
	&\times \exp \left( \Lambda_{10} \tilde \Lambda  - {\lambda _{{\rm{RE}}}}{{\mathord{\buildrel{\lower3pt\hbox{$\scriptscriptstyle\frown$}} 
				\over \gamma } }_{th}} - \frac{{\sum\nolimits_{t = 1}^n \vartheta_3 {\gamma _{th}}\vartheta_2}}{\Psi } \right)  \notag\\
	&\times {\left( {\frac{{{\gamma _{th}}}}{\Psi }} \right)^{A_2 }} \sum\limits_{q = 0}^{A_2 } {\left( {\begin{array}{*{20}{c}}
				A_2\\
				q
		\end{array}} \right)}  \frac{{{\vartheta_2^{A_2 - q}}{{ {\tilde \Lambda } }^{\frac{{q + 3}}{2}}}}}{{\vartheta_7\vartheta_2{\gamma _{th}}}} \notag\\
	& \times \int_1^{ + \infty } {{{\left( t \right)}^{\frac{{q + 1}}{2}}}\exp \left( { - \Lambda_{10} t\tilde \Lambda  } \right)}  \notag\\ &\times {K_{ - q + 1}}\left( {2\sqrt {\frac{{\sum\nolimits_{t = 1}^n \vartheta_3 {\gamma _{th}}{\lambda _{{\rm{RE}}}}t\tilde \Lambda }}{\Psi }} } \right)dt.
\end{align*}
where $\Lambda_{10} \triangleq \frac{{\left[ {{\lambda _{{\rm{JE}}}} + \frac{{{\gamma _{th}}{\lambda _{{\rm{RE}}}}\Theta }}{\Phi }} \right] }}{{\vartheta_7\vartheta_2{\gamma _{th}}}}$. Next, we apply the Maclaurin series as follows:
\begin{align}
	\label{EQ38}
	&\exp \left( { - \frac{{\left[ {{\lambda _{{\rm{JE}}}} + \frac{{{\gamma _{th}}{\lambda _{{\rm{RE}}}}\Theta }}{\Phi }} \right]t\tilde \Lambda }}{{\vartheta_7\vartheta_2{\gamma _{th}}}}} \right) \notag\\ &= \sum\limits_{w = 0}^\infty  {\frac{{{{( - 1)}^w}{{\left( {\frac{{\left[ {{\lambda _{{\rm{JE}}}} + \frac{{{\gamma _{th}}{\lambda _{{\rm{RE}}}}\Theta }}{\Phi }} \right]\tilde \Lambda }}{{\vartheta_7\vartheta_2{\gamma _{th}}}}} \right)}^w}{t^w}}}{{w!}}}.
\end{align}
By substituting \eqref{EQ38} into \eqref{EQ37}, the IP can be re-calculated as 
\begin{align}
	\label{EQ39}
	{\rm{IP}} &= 2\sum\limits_{n = 1}^N \frac{{{{( - 1)}^{n + 1}}}}{{n!}} {\left( {\frac{{\sum\nolimits_{t = 1}^n \vartheta_3 {\gamma _{th}}}}{{\Psi {\lambda _{{\rm{RE}}}}}}} \right)^{\frac{{ - q + 1}}{2}}} \notag\\&\sum\limits_{}^ *  {{\alpha _{{n_t}}}\frac{{{{( - 1)}^{{k_t}}}{{\left( {1 - {m_{{n_t}}}} \right)}_{{k_t}}}{{\left( {{\delta _{{n_t}}}} \right)}^{{k_t}}}{\lambda _{{\rm{RE}}}}{\lambda _{{\rm{JE}}}}}}{{{k_t}!{p_t}!{\vartheta_3^{A_3}}}}} \notag\\
%\end{align}
%\begin{align}
	& \times \exp \left( \tilde \Lambda \Lambda_{10} - {\lambda _{{\rm{RE}}}}{{\mathord{\buildrel{\lower3pt\hbox{$\scriptscriptstyle\frown$}} 
				\over \gamma } }_{th}} - \frac{{\sum\nolimits_{t = 1}^n \vartheta_3 {\gamma _{th}}\vartheta_2}}{\Psi } \right) \notag\\
	& \times {\left( {\frac{{{\gamma _{th}}}}{\Psi }} \right)^{A_2 }} \sum\limits_{q = 0}^{A_2 } {\left( {\begin{array}{*{20}{c}}
				{\sum\limits_{t = 1}^n {{p_t}} }\\
				q
		\end{array}} \right)}  \frac{{{\vartheta_2^{A_2 - q}}{{\left( {\tilde \Lambda } \right)}^{\frac{{q + 3}}{2}}}}}{{\vartheta_7\vartheta_2{\gamma _{th}}}} \notag\\
	& \times \sum\limits_{w = 0}^\infty  \frac{{{{( - 1)}^w}{{\left( \Lambda_{10} \tilde \Lambda  \right)}^w}}}{{w!}}  \int_1^{ + \infty } {{{\left( t \right)}^{\frac{{q + 1 + 2w}}{2}}}} \notag\\ & \times {K_{ - q + 1}}\left( {2\sqrt {\frac{{\sum\nolimits_{t = 1}^n \vartheta_3 {\gamma _{th}}{\lambda _{{\rm{RE}}}}t\tilde \Lambda }}{\Psi }} } \right)dt.
\end{align}
After computing the integral in \eqref{EQ39} and doing some algebra, we obtain the result as shown in the theorem.

\bibliographystyle{IEEEtran}
\bibliography{IEEEfull}

\vspace{-0.5cm}
\begin{IEEEbiographynophoto}
	{Tan N. Nguyen} (nguyennhattan@tdtu.edu.vn) was born in 1986 in Nha Trang City, Vietnam. He received a BS degree in electronics in 2008 from Ho Chi Minh University of Natural Sciences and an MS degree in telecommunications engineering in 2012 from Vietnam National University. He received a PhD in communications technologies in 2019 from the Faculty of Electrical Engineering and Computer Science at VSB – Technical University of Ostrava, Czech Republic. He joined the Faculty of Electrical and Electronics Engineering of Ton Duc Thang University, Vietnam, in 2013, and since then as been lecturing. His major interests are cooperative communications, cognitive radio, and physical layer security.
\end{IEEEbiographynophoto}
%%\vspace{0.5cm}
%%\vskip -70pt plus 4fil
%%\vspace{-1cm}
\begin{IEEEbiographynophoto}
	{Dinh-Hieu Tran} (S'20) was born in 1989 in Vietnam, growing up up in Gia Lai. He received a BE degree in Electronics and Telecommunication Engineering Department from Ho Chi Minh City University of Technology, Vietnam, in 2012. In 2017, he completed an MS degree with honours in Electronics and Computer Engineering at Hongik University, South Korea. He obtained a PhD at the Interdisciplinary Centre for Security, Reliability and Trust (SnT), University of Luxembourg, under the supervision of Prof. Symeon Chatzinotas and Prof. Bj\"orn Ottersten. His research interests include UAVs, IoTs, mobile edge computing, caching, backscatter, B5G for wireless communication networks. He is a recipient of the IS3C 2016 best paper award.
\end{IEEEbiographynophoto}
\begin{IEEEbiographynophoto} 
	{Trinh Van Chien} (S'16-M'20) received the B.S. degree in Electronics and Telecommunications from Hanoi University of Science and Technology (HUST), Vietnam, in 2012. He then received the M.S. degree in Electrical and Computer Enginneering from Sungkyunkwan University (SKKU), Korea, in 2014 and the Ph.D. degree in Communication Systems from Link\"oping University (LiU), Sweden, in 2020. He was  a research associate at University of Luxembourg. He is now with the School of Information and Communication Technology (SoICT), Hanoi University of Science and Technology (HUST), Vietnam. His interest lies in convex optimization problems and machine learning applications for wireless communications and image \& video processing. He was an IEEE wireless communications letters exemplary reviewer for 2016, 2017 and 2021. He also received the award of scientific excellence in the first year of the 5Gwireless project funded by European Union Horizon's 2020.
\end{IEEEbiographynophoto}
\begin{IEEEbiographynophoto}
	{Van-Duc Phan} (duc.pv@vlu.edu.vn) was born in 1975 in  Long An province, Vietnam. He received his MS degree from the Department of Electric, Electrical and Telecommunications Engineering at Ho Chi Minh City University of Transport in Vietnam and then in 201 a PhD from the Department of Mechanical and Automation Engineering, Da-Yeh University, Taiwan. His current research interests are sliding mode controls, non-linear systems and active magnetic bearings, flywheel energy storage systems, power system optimization, optimization algorithms, renewable energies, energy harvesting (EH) enabled cooperative networks, optical property improvement, lighting performance of white LEDs, energy efficiency LED driver integrated circuits, novel radio access technologies, and physical security in communications networks.
\end{IEEEbiographynophoto}
\begin{IEEEbiographynophoto}
	{Miroslav Voznak} (M'09-SM'16) received
his PhD in telecommunications in 2002 from the Faculty of Electrical Engineering and Computer Science at VSB – Technical University of Ostrava, and achieved habilitation in 2009. He was appointed Full Professor in Electronics and Communications Technologies in 2017. His research interests generally focus on ICT, especially quality of service and experience, network security, wireless networks, and big data analytics. He has authored and co-authored over one hundred articles indexed in SCI/SCIE journals. According to the Stanford University study released in 2020, he is one of the World’s Top 2\% of scientists in Networking \& Telecommunications and Information \& Communications Technologies. He served as a general chair of the $11^{th}$ IFIP Wireless and Mobile Networking Conference in 2018 and the $24^{th}$ IEEE/ACM International Symposium on Distributed Simulation and Real Time Applications in 2020. He participated in six projects funded by EU in programs managed directly by European Commission. Currently, he is a principal investigator in the research project QUANTUM5 funded by NATO, which focuses on the application of quantum cryptography in 5G campus networks.
\end{IEEEbiographynophoto}
\begin{IEEEbiographynophoto}
	{Symeon Chatzinotas}, (S'06-M'09-SM'13) is currently Full Professor / Chief Scientist I in Satellite Communications and Head of the SIGCOM Research Group at SnT, University of Luxembourg. He coordinates research activities in communications and networking, acting as a PI in over 20 projects and is the main representative for 3GPP, ETSI, DVB.
In the past, he worked as a Visiting Professor at the University of Parma, Italy, lecturing on 5G Wireless Networks. He was involved in numerous R$\&$D projects for NCSR Demokritos, CERTH Hellas and CCSR, University of Surrey.
He was co-recipient of the 2014 IEEE Distinguished Contributions to Satellite Communications Award and Best Paper Awards at EURASIP JWCN, CROWNCOM, ICSSC. He has (co-)authored more than 450 technical papers in refereed international journals, conferences and scientific books.
He is currently on the editorial board of the IEEE Transactions on Communications, IEEE Open Journal of Vehicular Technology, and the International Journal of Satellite Communications and Networking. 
\end{IEEEbiographynophoto}
\end{document}